\documentclass[aps,prx,twocolumn,superscriptaddress,longbibliography]{revtex4-1}

\usepackage[margin=1in]{geometry}

\usepackage{amsmath,mathtools,booktabs} 
\usepackage{amsthm}
\usepackage{mathrsfs,amssymb,amsfonts,bbold}  
\usepackage{stmaryrd}
\usepackage{accents}
\usepackage{bm}
\usepackage[mathscr]{euscript}
\usepackage{enumitem}
\usepackage[linkcolor=blue,colorlinks=true]{hyperref}
\usepackage{graphicx}
\usepackage[caption=false]{subfig}
\graphicspath{
    {figures/}
}

\usepackage{tikz}
\usetikzlibrary{calc}
\newcommand{\tm}[1]{\tikz[overlay,remember picture] \node (#1) {};}
\newcommand{\DrawBox}[5][]{%
    \tikz[overlay,remember picture]{%
        \coordinate (TopLeft)     at ($(#2)+(-#4em,0.8em)$);
        \coordinate (BottomRight) at ($(#3)+(#5em,-0.3em)$);
        \path (TopLeft); \pgfgetlastxy{\XCoord}{\IgnoreCoord};
        \path (BottomRight); \pgfgetlastxy{\IgnoreCoord}{\YCoord};
        \coordinate (LabelPoint) at ($(\XCoord,\YCoord)!0.5!(BottomRight)$);
        \draw [red,#1] (TopLeft) rectangle (BottomRight);
    }
}

\def \k {\bm{\mathrm{k}}}
\def \q {\bm{\mathrm{q}}}
\def \rr {\bm{\mathrm{r}}}
\def \lz {M_z}
\def \lxy {M_{xy}}
\def \lxyg {\tilde M_{xy}}

\newcommand{\sT}{\mathcal{T}}
\newcommand{\sP}{\mathcal{P}}
\newcommand{\soc}{\varepsilon_\text{soc}}
\newcommand{\hy}{\varepsilon_h}
\newcommand{\sro}{Sr$_2$RuO$_4$}
\newcommand{\sO}{\mathcal{O}}
\newcommand{\sE}{\mathcal{E}}
\newcommand{\sH}{\mathcal{H}}
\newcommand{\sV}{\mathcal{V}}
\newcommand{\sM}{\mathcal{M}}
\newcommand{\sC}{\mathcal{C}}

\newcommand{\sU}{\mathcal{U}}
\newcommand{\up}{\uparrow}
\newcommand{\down}{\downarrow}
\newcommand{\bra}{\langle}
\newcommand{\ket}{\rangle}
\newcommand{\ve}{\varepsilon}

\DeclareRobustCommand{\S}{S}
\DeclareRobustCommand{\D}{D}
\DeclareRobustCommand{\G}{G}
\def \hs {S}
\def \hd {D}
\def \hg {G}

\DeclareMathOperator{\tr}{tr}

\newtheorem{theorem}{Theorem}[section]

\newtheorem{lemma}[theorem]{Lemma}

\begin{document}

\title{Multiband mean-field theory of the $d+ig$ superconductivity scenario in \sro{}}
\author{Andrew C. Yuan}
\affiliation{Department of Physics, Stanford University, Stanford, CA 93405, USA}
\author{Erez Berg}
\affiliation{Department of Condensed Matter Physics, Weizmann Institute of Science, Rehovot 7610001, Israel}
\author{Steven A. Kivelson}
\affiliation{Department of Physics, Stanford University, Stanford, CA 93405, USA}

\begin{abstract}
  Many seemingly contradictory experimental findings concerning the superconducting state in \sro{} can be accounted for 
  on the basis of a conjectured  accidental degeneracy between two patterns of  pairing  that are unrelated to each other under the  $(D_{4h})$ symmetry of the crystal:  a $d_{x^2-y^2}$-wave $(B_{1g})$ and a $g_{xy(x^2-y^2)}$-wave $(A_{2g})$ superconducting state. 
  In this paper, we propose a generic multi-band model in which the $g$-wave pairing involving the $xz$ and $yz$ orbitals arises from second-nearest-neighbor BCS channel effective interactions.
  Even if time-reversal symmetry is broken in a $d+ig$ state, such a superconductor remains gapless with a Bogoliubov Fermi surface that approximates a (vertical) line node.   
  The model gives rise to a strain-dependent splitting between the critical temperature $T_c$ and the time-reversal symmetry-breaking temperature $T_\text{trsb}$ that is qualitatively similar to some of the experimental observations in \sro{}. 
\end{abstract}
\maketitle
\section{Introduction}

For more than two decades, \sro{} was generally believed to be a chiral $p$-wave superconductor (SC) mainly due to a compelling narrative based on early experiments \cite{mackenzie2003superconductivity,sigrist2005review,maeno2011evaluation,kallin2012chiral,mackenzie2017even}.
The extreme sensitivity of the superconducting state to impurities \cite{mackenzie1998extremely} unambiguously establishes that it is an unconventional SC, and various experiments confirm that the gap is nodal \cite{lupien2001ultrasound, bonalde2000temperature}.
However, recent nuclear magnetic resonance (NMR) experiments have seemingly ruled out triplet pairings of any sort \cite{ishida2020reduction,pustogow2019constraints,chronister2020evidence}.

Further constraints on the symmetry of the SC order can be inferred from a variety of experiments.  
Elasto-caloric data \cite{li2022elastocaloric} obtained when the Fermi surface is tuned through a Lifshitz transition strongly suggests that  the SC gap is non-vanishing at the Van Hove points.
Moreover, ultrasound experiments, taken at face value, suggest a two-component order parameter with highly constrained symmetries \cite{ghosh2021thermodynamic, benhabib2021ultrasound}:
A discontinuity in the $c_{66}$ shear modulus implies that the only possibilities consist of the innately two-dimensional irrep - $d_{xz}$ - $d_{yz}$ ($E_g$) \cite{vzutic2005phase, ramires,ramires2022nodal,suh2020stabilizing,ramires2019superconducting, kaser2022interorbital,romer2022leading} 
-  or an accidental degeneracy between two distinct one-dimensional irreps whose tensor product has $B_{2g}$ symmetry, i.e., $d_{x^2-y^2}$ \& $g$ \cite{kivelson2020proposal,yuan2021strain,sheng2022multipole,gingras2022superconductivity,clepkens2021shadowed} or $s$ \& $d_{xy}$ \cite{romer2020theory, romer2021superconducting}.
Evidence \cite{grinenko2021split,xia2006high} of  time-reversal symmetry (TRS) breaking at or slightly below $T_c$ provides further evidence of a two component order parameter.

That the Fermi surface of \sro{} is quasi-two-dimensional, i.e. it consists of three cylindrical  sheets  corresponding to the $\alpha$, $\beta$, and $\gamma$ bands, makes it seem unlikely that the SC order parameter has strong inter-layer pairing in the vertical $z$ direction, which provides an additional strong reason to exclude the possibility of $d_{xz}$ - $d_{yz}$ ($E_g)$ pairing.  
This is further supported by the lack of a visible discontinuity in the $(c_{11}-c_{12})/2$ ($B_{1g}$) shear modulus \cite{ghosh2021thermodynamic}.
Conversely, an accidental degeneracy between distinct irreps requires a certain degree of fine-tuning so that the critical temperature $T_c$ is roughly equal to the TRS breaking temperature $T_\text{trsb}$, i.e., $T_c\approx T_\text{trsb}$.
However, this requires only one-degree of fine tuning and thus could plausibly arise in a small subset of SC materials.

In our recent works \cite{kivelson2020proposal,yuan2021strain}, we preferred the $d_{x^2-y^2}+ig$ pairing symmetry over the $s+id_{xy}$ wave, mainly due to experimental observations of line nodes in the SC gap function \cite{hassinger2017vertical,lupien2001ultrasound, bonalde2000temperature}. 
Indeed, in a single-band model, $d+ig$ pairing leaves symmetry protected line nodes along the diagonal directions (110), (1$\bar{1}$0), while
$s+id_{xy}$ pairing would require an extra degree of fine tuning, i.e., a $s+id_{xy}$-wave SC is generically fully gapped.
However, a perturbative study of the effective interaction based on a ``realistic'' multi-band model of the electronic structure of \sro{} concluded that the leading instabilities  are in the $s$ and $d_{xy}$-channels \cite{romer2020theory, romer2021superconducting}; 
this result raises issues concerning the ``naturalness'' of the $d_{x^2-y^2}+ig$ SC state.
In particular, within a single band on a square lattice, the pair-wave-function in a $g$-wave state vanishes at all distances shorter than 4th nearest-neighbors, and thus seemingly requires unnaturally long-range interactions. 

To address this problem, we consider a generic multi-band microscopic BCS model in which the results follow largely from symmetry considerations.
In this model, the $g$-wave always lives primarily on the $\alpha$ and $\beta$ bands, which derive from the symmetry-related $xz$ and $yz$ Ru orbitals.
There are two reasons for this:
(1) $g$-wave pairing is strongly disfavored in the $xy$-band 
in that the pair-wave-function vanishes where the density of states is largest in the vicinity of the Van Hove points.
(2) 
Coupling between the $xz$ and $yz$ bands permits $g$-wave pairing to be induced by 2$^\text{nd}$ nearest-neighbor effective interactions \footnote{
Note:  The interactions here should be viewed as effective interactions.  However, so long as any  fluctuations that mediate such interactions are far from critical, the resulting effective interactions will be short-ranged in space and time - as assumed here.  For instance, the correlation length associated with short-range spin-density wave correlations ($\Delta q=0.13 ~\textup{\AA}^{-1}$)  seen in neutron scattering \cite{sidis1999evidence} has a correlation length that is $1/(a \Delta q)\approx 2$ lattice constants, where $a=3.85 ~\textup{\AA}$ is the in-plane lattice spacing.}, i.e. much shorter-range than is required in a single-band context.
Conversely, the $d_{x^2-y^2}$ component is strongest on the $xy$ band.

As a consequence, the model readily accounts for the recently observed   uniaxial stress dependence of the splitting of the critical temperature $T_c$ from the time-reversal symmetry (TRS) breaking $T_\text{trsb}$ 
\cite{grinenko2021split}.
As we show via microscopic BdG calculations (Fig. \ref{fig:TK-uniaxial}), application of uniaxial stress has little effect on $T_c$ until it is sufficiently strong to trigger a sharply peaked enhancement close to the critical strain at which the $xy$-band crosses a van Hove point, leading to a divergent  density of states (DOS) \cite{sunko2019direct,hsu2016manipulating,li2022elastocaloric}.
In contrast, across this Lifshitz transition, the $xz$ and $yz$-bands are only slightly distorted (an asymmetry of $\sim2\%$) \cite{grinenko2021split} and thus the $g$-wave component of the order parameter is 
largely unaffected, resulting in a $T_\text{trsb}$ that is only weakly strain dependent. 
(The strain dependence would be quite different for a putative $s+id_{xy}$-wave state, in which both components  sit on all bands.) 
Within our model, the approximate degeneracy between the $d$ and $g$ wave components persists in the presence of symmetry breaking shear strain. However, it is lifted by isotropic strain and leads to a linear increase in $T_c$; this effect has not been detected in recent experiments \cite{jerzembeck2022superconductivity}. How to reconcile this observation with our $d+ig$ proposal remains unresolved at present.

Regarding the nodal structure, we show that the symmetry-protected line nodes along the diagonal planes $(110),(1\bar{1}0)$ anticipated in the single-band model extend into Bogoliubov Fermi Surfaces (FS) in the general multi-band $d_{x^2-y^2}+ig$ state.
The Bogoliubov FSs are extended in the $k_z$ direction and
consist of narrow ellipses for generic $k_z$, that pinch into point nodes in the $k_z=0$ and $\pi$ planes.
The pinching off at $k_z=0\: \& \: \pi$ is due to a combination of the mirror symmetries $\sM_z: z\mapsto -z$ and $\sM_{xy}: x\leftrightarrow y$ \footnote{In general, mirror symmetry $\sM_z$ acts on both on $\k$-space and orbital/spin space $\lz$ so that $\sM_z = \lz \otimes (k_z\mapsto -k_z)$.
In the $k_z=0,\pi$-planes, $\sM_z$ conserves $\k$ so that $\sM_z =\lz$, and thus when the state is even under $\sM_z$, the $\k$-conserving Hamiltonian is decoupled in the eigenspaces $\lz=\pm i$. Away from the $k_z=0,\pi$ planes, such a decomposition is not exact and thus the radius of the Bogoliubov Fermi surface is in proportion to the weak coupling between the orbital/spin eigenspaces $\lz=\pm i$. This is not necessarily on the same magnitude as the inter-plane coupling. Indeed, one could imagine a general 3D Hamiltonian with $k_z$ dependence (due to inter-plane coupling), but still decoupled in the eigenspaces $\lz=\pm i$ involving only orbital/spin (see Appendix \eqref{sec:stability-details}. This particular scenario would not be enough to induce Bogoliubov Fermi surfaces.}.
In practice, such Bogoliubov Fermi surfaces would exhibit behavior characteristic of line nodes (as is observed experimentally \cite{hassinger2017vertical,lupien2001ultrasound, bonalde2000temperature}), except at extremely low temperatures.

\section{Microscopic  Analysis}
\begin{figure}[ht]
\subfloat[\label{fig:FS-a}]{%
  \centering
  \includegraphics[width=.48\columnwidth]{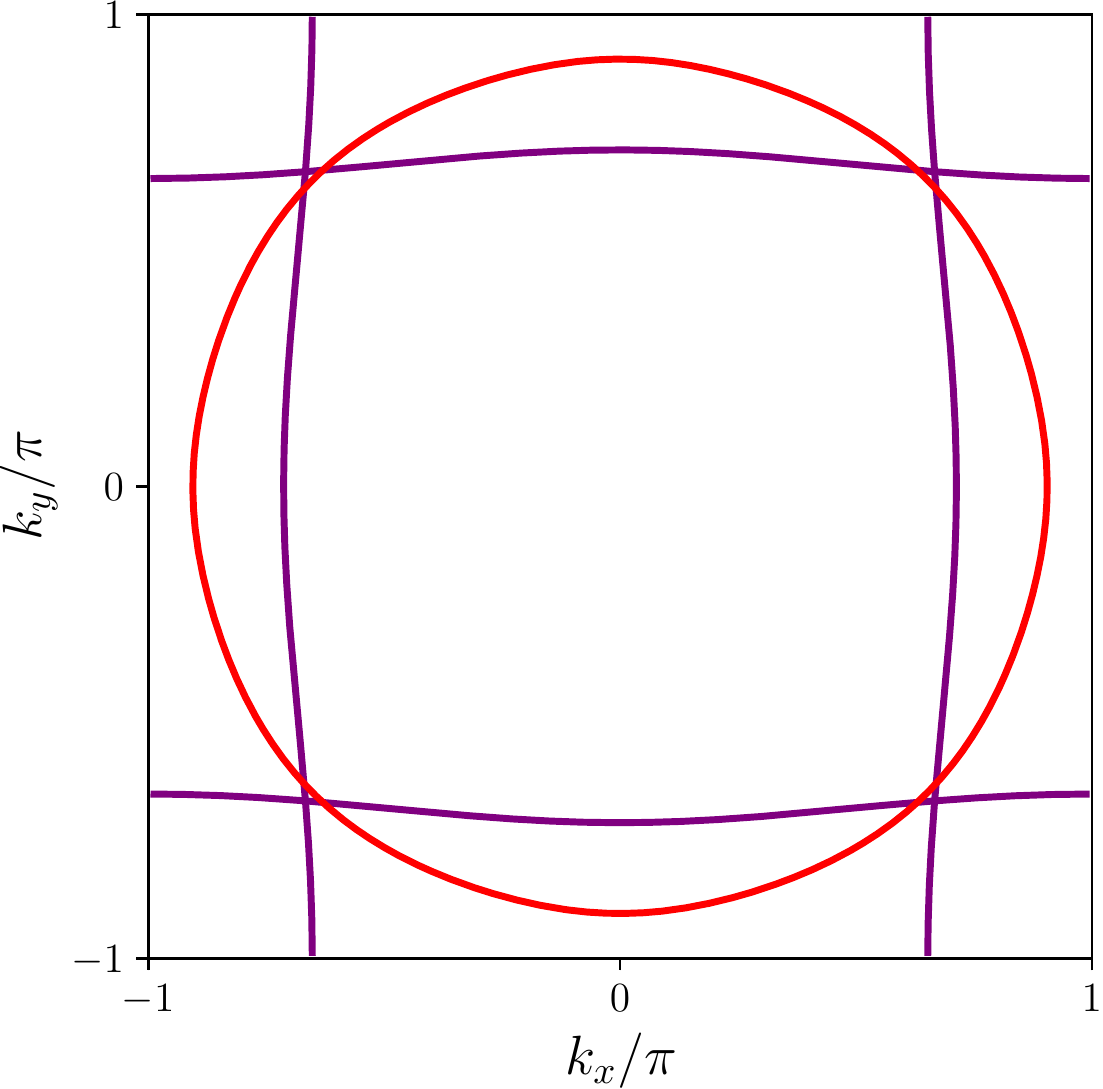}
}
\subfloat[\label{fig:FS-b}]{%
  \centering
  \includegraphics[width=.48\columnwidth]{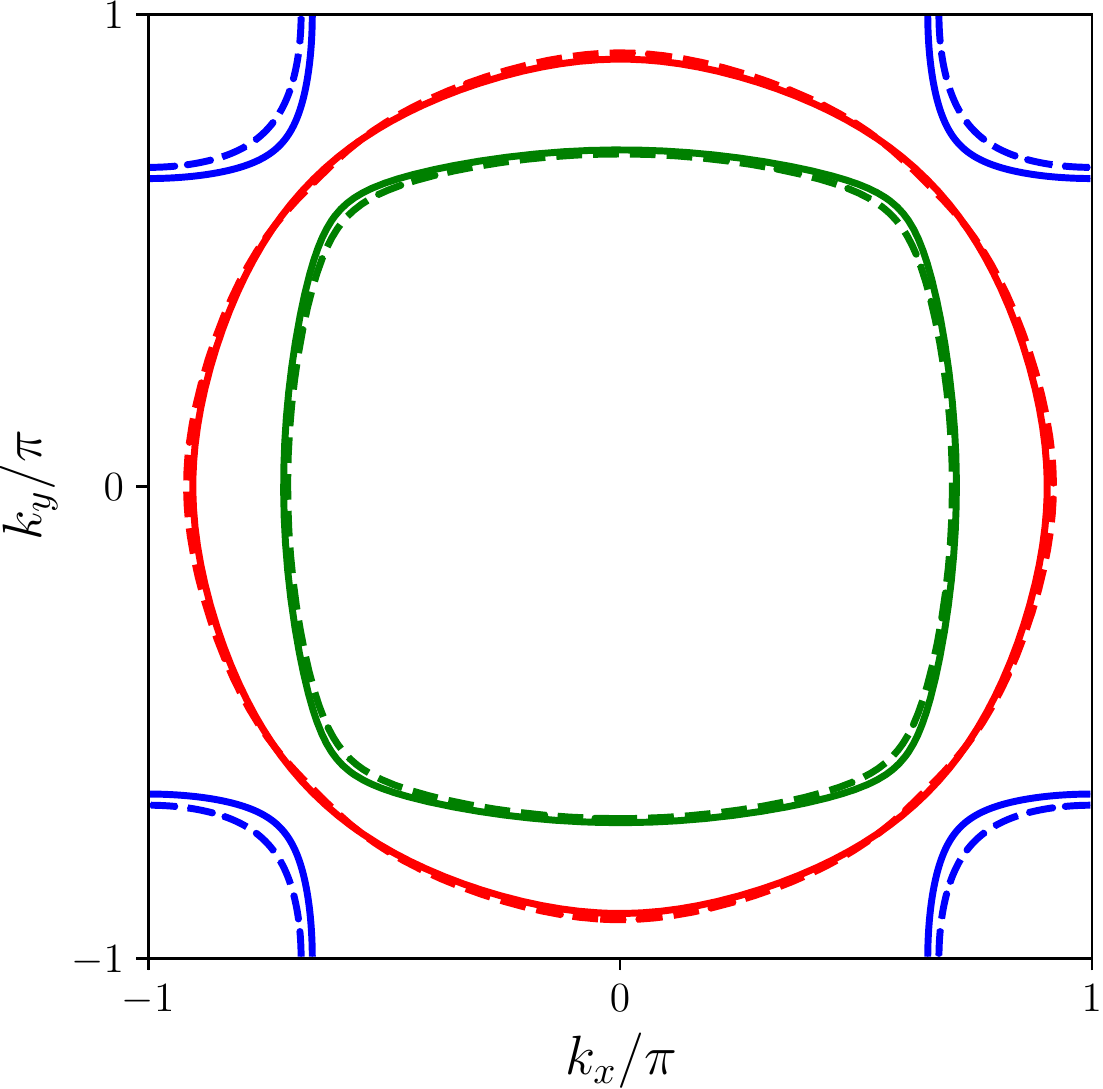}
}
\caption{(a) Fermi surfaces (FS) computed in the absence of hybridization and spin-orbit (SO) coupling. (b) FS including the effect of hybridization between the $xz,yz$ bands but no SO (solid lines), and with both hybridization and SO (dashed lines). Here interplane dispersion is neglected; numerical parameters are given in Appendix \eqref{sec:band-structure}. 
}
\label{fig:FS}
\end{figure}
\subsection{Setup}

Unless otherwise specified, we will consider a 2D model - corresponding to a single RuO$_2$ plane (We  we will consider the effects of 3D dispersion in the discussion of the nodal structure).  
We will always restrict our attention (for reasons already discussed) to gap functions that are even with respect to both parity, $\sP$, and under mirror symmetry $\sM_z$.
This leaves us with only one-dimensional irreps of the point group $D_{4h}$, i.e., $s,d,d',g$-wave pairings ($A_{1g},B_{1g},B_{2g},A_{2g}$, respectively).

\subsection{Normal State}

Let us first consider the normal state $\Delta =0$ where the Hamiltonian is
\begin{align}
    \label{eq:normal}
    \sH_0 &= \sH_\text{intra}  +\sH_\text{hybrid} +\sH_\text{SO} \\
    \sH_\text{intra} &= \sum_{\k\nu s} \ve_{\nu}(\k) \psi_{\nu s}^\dagger(\k) \psi_{\nu s}(\k) \nonumber\\
    \sH_\text{hybrid} &= \sum_{\k s} \hy(\k) \left(\psi_{xs}^\dagger(\k) \psi_{ys}(\k) + \text{h.c} \right) \nonumber\\
    \sH_\text{SOC} &= \sum_{\k\nu\nu'ss'} \soc(\k) \psi_{\nu s}^\dagger (\k)\psi_{\nu's'}(\k) \bm{l}_{\nu\nu'}\cdot \bm{\sigma}_{ss'} \nonumber
\end{align}
The first term $\sH_\text{intra}$ represents the dispersion of each orbital $\nu= x,y,z$ (for $d_{yz}, d_{zx}, d_{xy}$ respectively), while the second term $\sH_\text{hybrid}$ is the spin-conserving hybridization between the $x,y$ bands and generates the $\alpha,\beta$ bands shown in Fig. \ref{fig:FS}, where $\hy(\k)$ characterizes the strength of hybridization between the $x,y$ bands 
\footnote{It should be noted that due to $\sM_z$ symmetry, there exists no spin-conserving hybridization between the $x,z$-bands and $y,z$-bands, i.e., $\ve_{x,z}=\ve_{y,z}=0$. However, in a full 3D model, such hybridization terms are allowed and can give rise to gap functions between the $x,z$-bands and $y,z$-bands, i.e., $\Delta_{x,z},\Delta_{y,z}\ne0$, even in the absence of spin-orbit coupling. Since \sro{} is quasi-2D, such terms can always be treated perturbatively.}.
The third term represents the spin-orbit coupling (SOC) where $\bm{l}^\mu_{\nu\nu'} = -i\epsilon_{\nu\nu'\mu}$.
The specific parametrizations of $\ve_\nu(\k)$, $\hy(\k)$, and $\soc(\k)$ used in our calculations are given in the Appendix \eqref{sec:band-structure}. 

In the normal state, $\sP,\sM_z,\sT$ are all preserved. 
The single particle space can thus be decomposed into the eigenspaces of $\sM_z =\pm i$ \footnote{That $\sM_z^2 = -1$ is due to its action on the spin sector.}  with ordered basis $|+\ket\equiv (x{\up},y{\up}, z{\down})$ and $|-\ket \equiv (x{\down}, y{\down}, -z{\up})$, respectively \footnote{Technically, $|\pm \ket$ are ordered bases for the eigenspaces of $\lz=\pm i$ where $\lz$ is the orbital/spin action of $\sM_z$ so that $\sM_z =\lz\otimes (k_z\mapsto -k_z)$. However, since we are dealing with an ideal 2D model, the distinction will not be important until we consider a general 3D model (see Appendix \eqref{sec:stability-details}).}. 
The signs are chosen so that the 
$\k$-conserving anti-unitary symmetry $\sT \sP$ maps $|+\ket \mapsto |-\ket$, and thus the two eigenspaces have equivalent energy levels 
\footnote{It should be mentioned that even in the absence of mirror symmetry $\sM_z$, the anti-unitary symmetry $\sT\sP$ is sufficient to show that energy levels are doubly degenerate (and thus there are only 3 Fermi surfaces).
The potential benefit of this extra condition is to obtain a canonical decoupling which will persist into the SC state when TRS is broken.}.
The normal state Hamiltonian is simplified as
\begin{align}
    \label{eq:normal-simplified}
    \sH_0 &= \sum_{\k,a=\pm} \psi_a^\dagger (\k) \hat{\ve}_a (\k)\psi_a(\k) \\
    \hat{\ve}_{\pm} &= \bra \pm |\ve |\pm' \ket = 
    \begin{bmatrix}
    \ve_{x}      & \hy\mp i\soc & \soc     \\
    \hy\pm i\soc  & \ve_{y}     & \mp i\soc \\
    \soc          & \pm i\soc    & \ve_{z}
    \end{bmatrix} \nonumber
\end{align}
Where $\psi_{\pm}^\dagger$ are the creation operators for the basis $|\pm \ket$ and $\hat{\ve}_- = \hat{\ve}_+^*$ due to $\sT\sP$ \footnote{The apostrophe in $|\pm '\ket$ is to denote a possible different vector (within the same $\pm$ sector) than $|\pm \ket$, e.g., $\bra x\up| \ve |z\down\ket$.}. 
Using a multi-orbital symmetry analysis (Table II of \cite{ramires2019superconducting} or Appendix \eqref{sec:symmetry-operations}), the components $\ve_x(\k) \pm \ve_y(\k),\hy(\k)$ must satisfy $S,D,D'$-wave ($A_{1g},B_{1g},B_{2g}$) symmetry in $\k$-space, respectively \footnote{$S,D,D'$ are capitalized to emphasize that the symmetries are only with respect to $\k$-space and not orbitals/spin. It should be noted that the symmetries of each component were obtained assuming that the normal state $\ve(\k)$ is invariant under the point group $D_{4h}$. If, for example, $100$- or $110$- strain was included, the point group would shrink (though the normal state would still preserves mirror symmetry $\sM_z$).}.

\subsection{BCS State}
Since the BCS Hamiltonian is assumed to preserve $\sM_z$, a similar decomposition can be performed on the general multi-band particle-hole space.
More specifically, particle states $|+\ket$ and hole states $|-\ket$ form an ordered basis for the eigenspace $\sM_z=+i$, and particle states $|-\ket$ and hole states $|+\ket$ for $\sM_z = -i$. 
Let $\sC$ denote particle-hole symmetry so that $\sC\sP$ denotes a 
$\k$-conserving anti-unitary symmetry  (maps $\sH(\k) \mapsto -\sH(\k)$) between the $\sM_z = +i \leftrightarrow  -i$ eigenspaces.
Hence, the BCS Hamiltonian $\sH$ in the decoupled eigenspaces $\sM_z = \pm i$ have equivalent (opposite signed) energy levels.
When TRS $\sT$ is preserved, each eigenspace is further doubly (opposite signed) degenerate due to the local unitary symmetry $\sT\sC$ which commutes with $\sM_z$ \footnote{Even in the absence of mirror symmetry $\sM_z$, the symmetries $\sT\sP,\sC\sP$ are sufficient to show that every nonzero eigen-energy is doubly degenerate, while the zero eigen-energy space is quartic-degenerate. However, the mirror symmetry $\sM_z$ will make it possible to analyze nodal points of TRS breaking accidentally degenerate states, e.g., $d+ig$ pairings.}.
Therefore, we can use the Nambu spinors $\Psi_\pm^\dagger(\k) = (\psi_\pm ^\dagger(\k), \psi_\mp(-\k))$ to write the BCS Hamiltonian in the following block-matrix form
\begin{align}
    \label{eq:BCS-block}
    \sH &= \sum_{\k,a=\pm}  \Psi_a^\dagger(\k) H_{a}(\k) \Psi_a(\k)\\
    H_{\pm}(\k) &=
    \begin{bmatrix}
    \hat\ve_\pm(\k) & \hat\Delta_\pm(\k)\\
    \hat\Delta_\pm^\dagger(\k) & -\hat\ve_\pm(\k)
    \end{bmatrix} \nonumber\\
    \hat{\Delta}_+(\k) &= 
    \begin{bmatrix*}[c]
      \Delta_{x}    & \Delta_{x,y}      & \Delta_{x,z} \\
      \Delta_{y,x}  & \Delta_{y}        & \Delta_{y,z} \\
      \Delta_{z,x}  & \Delta_{z,y}      & \Delta_{z}
    \end{bmatrix*}, \quad \hat{\Delta}_- = -\hat{\Delta}_+^T
    \nonumber
\end{align}
where $\hat{\ve}_\pm (\k)$ is given in Eq. \eqref{eq:normal-simplified}. 
In the case where TRS $\sT$ is preserved, $\hat{\Delta} (\k)^\dagger =\hat{\Delta} (\k)$ \footnote{Due to the symmetry between $\hat{\Delta}_\pm$, we will drop the $+$ subscript when referring to the gap matrix $\hat{\Delta}_+$.}.
\subsection{Symmetries of Gap Matrix}
\begin{table}[h!]
\centering
\begin{tabular}{||c|c|c||}
 \hline
 Component $C$ \textbackslash  Operation $\sO$    & $\sC_4$ & $\sM_{xy}$ \\ [0.5ex]
 \hline\hline
 $\Delta_z$ & $+$ & $+$ \\\hline
 $\Delta_x+\Delta_y$ & $+$ & $+$ \\\hline
 $\Delta_x-\Delta_y$ & $-$ & $-$ \\\hline
 $\Delta_{x,y}+\Delta_{y,x}$ & $-$ & $+$ \\\hline\hline 
 $\Delta_{x,y}-\Delta_{y,x}$ & $+$ & $+$ \\\hline 
 $\Delta_{x,z}+\Delta_{z,x}+i(\Delta_{y,z}-\Delta_{z,y})$ & $+$ & $+$\\\hline
 $\Delta_{x,z}+\Delta_{z,x}-i(\Delta_{y,z}-\Delta_{z,y})$ & $-$ & $-$\\\hline
 $\Delta_{y,z}+\Delta_{z,y}+i(\Delta_{x,z}-\Delta_{z,x})$ & $-$ & $+$\\\hline
 $\Delta_{y,z}+\Delta_{z,y}-i(\Delta_{x,z}-\Delta_{z,x})$ & $+$ & $-$\\\hline
\end{tabular}
\caption{\label{tab:symmetry-mapping}  The $\eta = \pm$ signs in the table are such that $C(\k) \mapsto \eta C(\k')$ under point group operation $\sO=\sC_4,\sM_{xy}$  where $\k' \equiv \sO^{-1} \k$. The first 4 components (last 5) correspond to singlet (triplet) pairing.}
\end{table}

Using a multi-orbital symmetry analysis (\cite{ramires2019superconducting} or Appendix \eqref{sec:symmetry-operations}), it is straightforward (but a bit complicated) to see that the matrix elements of $\hat{\Delta}(\k)$ are mapped under $\pi/2$-rotation $\sC_4$ and $x\leftrightarrow y$ mirror symmetry $\sM_{xy}$ (the remaining generators of $D_{4h}$), in the manner tabulated in Table. \ref{tab:symmetry-mapping}.
Depending on the overall symmetry of the SC gap, the remaining matrix components must satisfy corresponding symmetries in $\k$-space shown in Table \ref{tab:symmetry}.

\begin{table}[h!]
\centering
\begin{tabular}{||c | r| r| r| r||}
 \hline
  Component $C\backslash$  Symmetry     & $s$     & $d$     & $d'$   & $g$ \\ [0.5ex]
 \hline\hline
 $\Delta_z$                  & $\S$     & $\D$    & $\D'$  & $\G$  \\\hline
 $\Delta_x+\Delta_y$         & $\S$     & $\D$    & $\D'$  & $\G$ \\\hline
 $\Delta_x-\Delta_y$         & $\D$     & $\S$    & $\G$   & $\D'$ \\\hline
 $\Delta_{x,y}+\Delta_{y,x}$ & $\D'$    & $\G$    & $\S$   & $\D$ \\
 \hline\hline 
 $\Delta_{x,y}-\Delta_{y,x}$                              & $\S$     & $\D$    & $\D'$  & $\G$ \\\hline 
 $\Delta_{x,z}+\Delta_{z,x}+i(\Delta_{y,z}-\Delta_{z,y})$ & $\S$     & $\D$    & $\D'$  & $\G$\\\hline
 $\Delta_{x,z}+\Delta_{z,x}-i(\Delta_{y,z}-\Delta_{z,y})$ & $\D$     & $\S$    & $\G$   & $\D'$\\\hline
 $\Delta_{y,z}+\Delta_{z,y}+i(\Delta_{x,z}-\Delta_{z,x})$ & $\D'$    & $\G$    & $\S$   & $\D$\\\hline
 $\Delta_{y,z}+\Delta_{z,y}-i(\Delta_{x,z}-\Delta_{z,x})$ & $\G$     & $\D'$   & $\D$   & $\S$\\\hline
 \hline
\end{tabular}
\caption{\label{tab:symmetry} $\k$-space symmetry of components in gap matrix $\hat{\Delta}(\k)$ in Eq. \eqref{eq:BCS-block}. The top row denotes the overall symmetry of the SC gap, while each entry denotes the $\k$-space symmetry of the corresponding component, capitalized ($S,D,D',G$) to emphasize $\k$-space symmetry only. The first 4 components (last 5) correspond to singlet (triplet) pairing.}
\end{table}

It is worth mentioning that SOC in general mixes (physical) spin singlet- and triplet-pairings. Indeed, the gap matrix $\hat{\Delta}=\hat{\Delta}^\text{singlet}+\hat{\Delta}^\text{triplet}$ can be decomposed into a singlet component $\hat{\Delta}^\text{singlet}$ and a triplet component $\hat{\Delta}^\text{triplet}$ (referred to as \textit{shadowed triplet} in \cite{clepkens2021shadowed}), where
\begin{align}
    \hat{\Delta}^\text{singlet}(\k) &=
    \begin{bmatrix*}[c]
      \Delta_{x}    & \Delta_{h}      & 0 \\
      \Delta_{h}    & \Delta_{y}      & 0 \\
      0             & 0               & \Delta_{z}
    \end{bmatrix*}\\
    \hat{\Delta}^\text{triplet}(\k) &=
    \begin{bmatrix*}[c]
      0              & i\Delta_{h'}      & \Delta_{x,z} \\
      -i\Delta_{h'}   & 0                  & \Delta_{y,z} \\
      \Delta_{z,x}   & \Delta_{z,y}      & 0
    \end{bmatrix*}\\
    \Delta_{h}  &= \frac{1}{2}(\Delta_{x,y} + \Delta_{y,x}) \\
    \Delta_{h'} &= \frac{1}{2i}(\Delta_{x,y} - \Delta_{y,x})
\end{align}

We emphasize that we do not assume SOC is weak; rather the dominant singlet behavior is a phenomenological assumption.
In this case, since $G$ waves require unnaturally long-range interactions to generate, Table \ref{tab:symmetry} implies that the $g$ wave ``naturally" lives mainly on the $x,y$-bands, while all other 1D irreps $s,d,d'$ wave should have comparable magnitudes on all three bands. 
More specifically, a $g$ wave can be generated via the form $\Delta_{x/y}(\k) = \pm D'(\k)$ which only requires 2$^\text{nd}$ nearest-neighbor interactions \footnote{Indeed, a state with $\Delta_{x,y}=\Delta_{y,x}=D$ has $g$ wave symmetry involving only nearest-neighbor pairing. However, since the hybridization $\hy$ in the band structure is generally small compared to the diagonal terms $\ve_x,\ve_y,\ve_z$, the off-diagonal components of the gap matrix $\hat{\Delta}$ are typically small for reasons unrelated to their range.}.

\subsection{Nodal Structure}

We now turn to the nodal structure of the dominantly singlet $d+ig$ phase. 
In a multi-band model, the nodal points are determined in a nonlinear manner by the SC gap $\Delta$ (i.e., by solving $\det H(\k) = 0$) and thus the nodal structure of accidentally degenerate states 
does not follow trivially from analyzing 1D irreps individually, e.g., \cite{ramires2019superconducting,ramires2022nodal}.
Typically, the stability of point nodes in 2D or line nodes in 3D is established under only under the assumption of unbroken TRS $\sT$ \cite{Blount1985,Berg2008,beri2010topologically,Sato2014}.

Below, we show that for the singlet $d+ig$ pairing in a  2D system, the existence/stability \cite{existence-stability} of point nodes along the $(110), (1\bar{1}0)$ directions can established, despite the breaking of TRS, on the basis of unbroken  $\sM_z$ (along with the usual $\sM_{xy}$)
\footnote{A similar proof can be adapted to show the existence/stability of line nodes along the $(100), (010)$ directions for $d'+ig$ pairing}.
A small admixture of subdominant triplet pairing due to SOC does not affect this result  (See Appendix \eqref{sec:stability-details}).
In a general 3D model, however, this argument only applies in the $k_z=0,\pi$ planes;  for all other (non-mirror invariant) values of $k_z$, the line nodes of the 2D problem expands into narrow surfaces of gapless Bogoliubov quasi-particles (similar to \cite{agterberg2017bogoliubov}), that become ``pinched'' into points at the $k_z = 0$ and $\pi$ planes. 
The radius of these ``Bogoliubov Fermi surfaces'' is small  in proportion to the coupling of eigenspaces $\sM_z = \pm i$ (See Appendix \eqref{sec:stability-details}).

In the \sro{} context, any such Bogoliubov Fermi surfaces would be small both as a consequence of the smallness of the spin-conserving hybridization between the $x,z$ and $y,z$ bands \footnote{Terms of the form $\bra xs|\ve(\k)|z s\ket,\bra y s|\ve (\k)|z s\ket$ for $s={\up},{\down}$, which break the action of mirror symmetry $\sM_z$ on the orbital/spin space (couple the $|\pm \ket$ bases).} and atypical forms of SOC \footnote{Terms such as $\bra x{\up}|\ve(\k)|x{\down} \ket,\bra x{\up}|\ve(\k)|y{\down} \ket$, since the typical form $\soc{}(\k) \bm{l\cdot\sigma}$ in Eq. \eqref{eq:normal} preserves $SU(2)$ action on the orbital/spin space and thus, even for $k_z \ne 0,\pi$, is decoupled in the orbital/spin bases $|\pm \ket$.} in units of the band-width. 
Thus, a multi-band $d+ig$ phase will display properties characteristic of vertical line nodes (such as a linear dependence of the density of states in energy) down to extremely low energies, consistent with the nodal behavior of the single band $d+ig$ phase.

\subsubsection{Existence of line nodes}
\label{sec:existence}

By symmetry considerations tabulated in Table \ref{tab:symmetry}, if the overall symmetry is of singlet $d+ig$-pairing, then along the $(110),(1\bar{1}0)$ directions, the only nonzero component of the gap matrix $\hat{\Delta}(\k)$ is $\Delta_x(\k)-\Delta_y(\k)$ (which we refer to as $2\Delta_3(\k)$ and possibly a complex number, i.e., $\Delta_3(\k) \equiv e^{i2\theta(\k)} |\Delta_3(\k)|$).
It's then clear that by performing a $\k$-conserving gauge transform, the quasi-particle Hamiltonian $H_+(\k)$ with complex $\Delta_3(\k)$ is mapped unitarily to that with nonnegative $|\Delta_3(\k)|$, i.e.,
\begin{equation}
    e^{-i\theta \tau_3} H_+(\hat{\ve}_+,\Delta_3) e^{i\theta \tau_3} =  H_+(\hat{\ve}_+,|\Delta_3|)
\end{equation}
Where $\tau_3$ denotes the Pauli matrix in particle-hole space and $H_+$ depends on $\k$ implicitly via $\hat{\ve}(\k),\Delta_3(\k)$. 
Equivalently, the gauge transform maps the TRS breaking $\k$-conserving Hamiltonian $H_+(\k)$ to such that preserves TRS \footnote{Indeed, gauge symmetry is the reason why single-band $d+ig$ states can still have stable line nodes despite breaking TRS} and thus a line node can be found explicitly along the $(110),(1\bar{1}0)$ directions (see Appendix \eqref{sec:nodal-point}).
It is worth mentioning that if we consider a triplet $d+ig$ pairing, by Table \ref{tab:symmetry-mapping} or \ref{tab:symmetry}, there are two nonzero components along the $(110),(1\bar{1}0)$ directions and thus the argument breaks down.

\section{Numerical Results}
\subsection{Setup}
\begin{table}[h!]
\centering
\begin{tabular}{||r r| r| r| r| r||}
 \hline
  Lattice   & Harmonic            & $R=0$    & $R=1$                 &  $R=2$\\ [0.5ex]
 \hline\hline
 $(A_{1g})$ &$\hs_R(\k)$           & $1$      & $\cos k_x +\cos k_y$  &  $2\cos k_x \cos k_y$\\
 \hline
 $(B_{1g})$ &$\hd_R(\k)$           & 0       & $\cos k_x -\cos k_y$  &  0\\
 \hline
 $(B_{2g})$ &$\hd'_R(\k)$          & 0        &  0                    & $2\sin k_x \sin k_y$ \\
 \hline
\end{tabular}
\caption{\label{tab:harmonics}First three harmonics ($R\le 2$) based on symmetry and normalized via their squared integral. The $0$ implies that there does not exist such a symmetry component for the given range. We have omitted the $\hg_R(\k)$-wave since it does not appear until 4th nearest-neighbor interactions $(R=4)$.}
\end{table}

Since on physical grounds, we plan to deal only with the case in which the effective interactions are short-ranged, it is convenient to introduce \textit{lattice harmonics}, specified by their range $R$ and their symmetry, which form an ordered basis for any lattice function.
The first few such harmonics for each relevant symmetry are shown in Table \ref{tab:harmonics}, normalized via their squared integral ($L^2$-norm).  
Thus, for example, we can decompose a $\D$ wave in $\k$-space into a linear combination of lattice harmonics $\hd_R (\k)$ of range $R$ (where $R=0$ indicates onsite, $R=1$ is nearest neighbors, $R=2$ is next-nearest neighbors, etc.).

We can then consider an effective interaction in the BCS channel of the form
\begin{align}
    \label{eq:interaction}
    \sV = \frac{1}{2}\sum_{\mu\mu',\k\k'} V_{\mu \mu'} (\k-\k') P^\dagger_{\mu}(\k) P_{\mu'} (\k')
\end{align}
where $P_\mu (\k) = \psi_{\down \mu}(-\k) \psi_{\up \mu}(\k)$ 
annihilates a  Cooper pair in band $\mu$ and $V_{\mu \mu'}(\q)$ represents the interaction between bands $\mu,\mu'$. 
By construction, the form of $V_{\mu \mu'}(\q)$ precludes interband pairing of the SC gap, i.e. the self-consistency equations imply $\Delta_h=0$;  a more general form of the interaction would generally lead to a small but non-zero $\Delta_h$, but this would not change any of our findings qualitatively.

We shall take interactions that respect the lattice symmetries and have range $R \le 2$ \footnote{By symmetry, it's possible that $V_{x,x}(\q) \ne V_{y,y}(\q)$ and $V_{x,z}(\q)\ne V_{y,z}(\q)$. However, the differences $(V_{x,x} -V_{y,y})(\q)$ and $(V_{x,z}-V_{y,z})(\q)$ must satisfy $B_{1g}$ symmetry with respect to $\q$. Therefore, for interaction range $R\le 2$, the interactions can only mix $s$ and $d_{x^2-y^2}$ wave harmonics, e.g., $(V_{x,x} - V_{y,y})(\k-\k') = s_1(\k) d_1(\k') + d_1(\k) s_1(\k') + \cdots $, and thus have no direct effect on the $R=2$ harmonics involving the $d_{xy}$ or $g$-waves.}:
\begin{equation}
    \label{eq:interaction-range}
  V_{\mu\mu'} (\q) = \sum_{R=0,1,2} v_{\mu\mu'}(R)  \hs_R (\q),
\end{equation}
where $\hs_R$ is given by the first row in Table \ref{tab:harmonics}.
It is then convenient to express $V$ as a sum of terms of the form $c_a\psi_a(\k) \psi_a(\k')$ where each $\psi_a(\k)$ transforms according to one of the irreps ($S,D,D'$) of the point group and decomposed into harmonics tabulate in Table \ref{tab:harmonics}.
We leave the details of the decomposition in Appendix \eqref{sec:interaction}.
For completeness, the nonlinear gap equation is given here
\begin{align}
    \label{eq:nonlinear-gap}
    \Delta_\mu(\k) = \sum_{\mu'\k'}V_{\mu\mu'}(\k-\k')\langle P_{\mu'}(\k')\rangle
\end{align}
where the expectation value $\bra\cdots \ket$ is taken with respect to the BCS Hamiltonian.

For simplicity, we will henceforth set SOC to zero, since it does not qualitatively affect the results. 
In particular, the Lifshitz transition occurs at $\k$-points $(\pi,0),(0,\pi)$, where the orbital character of each band is well-defined and thus the SOC terms are weak in any case.
In this limit,  the 2-band system consisting of the 
$\alpha$ and $\beta$ bands decouples from the $\gamma$ band within the BCS Hamiltonian, though all three bands are still coupled via the interaction term $\sV$ in the nonlinear gap equation in Eq. \eqref{eq:nonlinear-gap}.
Using the fact that $SU(4)$ is the spin group of $SO(6)$, we provide an intuitive manner of diagonalizing the $4\times 4$ BdG Hamiltonian involving the 2-band system, details of which are left in Appendix \eqref{sec:su4-so6}.

\subsection{Zero temperature $T=0$}
\begin{figure}
\centering
\includegraphics[width=1\columnwidth]{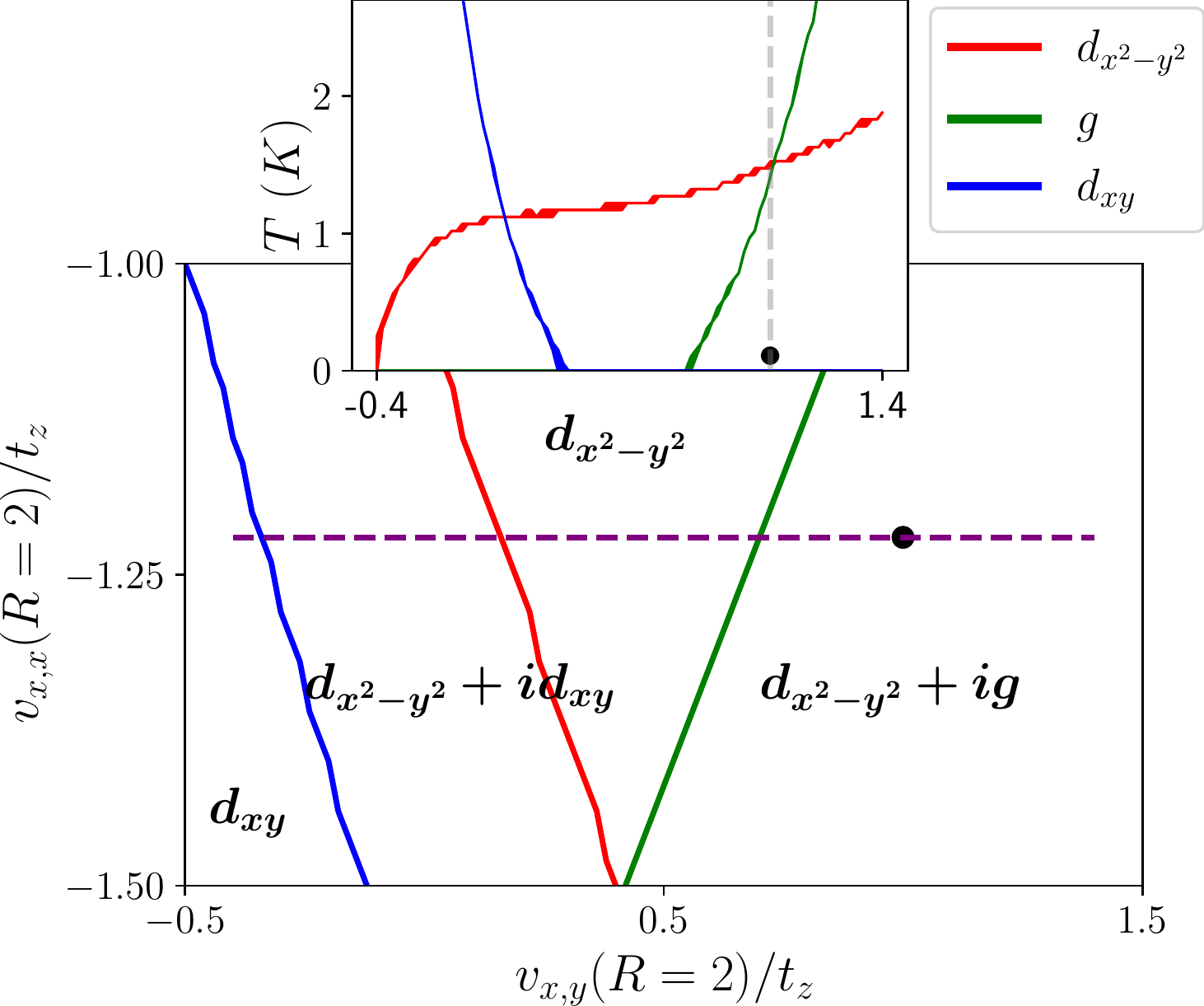}
\caption{
The $T=0$ phase diagram in a representative 2D plane in the multi-dimensional parameter space of interactions.  
The black point corresponds to a point at which $d$ and $g$ are degenerate, with  parameters listed in Table \ref{tab:interactions} - we will use these parameters in later analysis including that reported in Figs. \ref{fig:TK-uniaxial} and \ref{fig:pure-strain}.  
In the main figure, all interactions with range $R< 2$ are held fixed,  while those with $R= 2$ are varied maintaining the (arbitrarily chosen) relation $v_{z,z} = 0.73 v_{x,y} - 0.43$, where the slope is positive so that $v_{z,z},v_{x,y}$ do not compete \cite{slope}.
The inset represents the thermal phase diagram along the 1D cut through parameter space indicated by the dashed line in the main figure.
The different color solid lines are phase boundaries corresponding to the point at which the indicated symmetry order parameters vanish.
}
\label{fig:phase-diagram}
\end{figure}

\begin{table}
\centering
\begin{tabular}{||r| r| r| r| r||}
 \hline
  Interaction $v/t_z$     & $R=0$    & $R=1$                 &  $R=2$\\ [0.5ex]
 \hline\hline
 $v_{z,z}(R)$            & $7$      & $-0.4$  &  $0.3$\\
 \hline
 $v_{x,x}(R)$            & $5$      & $-0.1$  &  $-1.2$\\
 \hline
 $v_{x,y}(R)$            & $5$      & $-0.1$  &  $1$\\
 \hline
 $v_{x,z}(R)$            & $5$      & $0.1$   &  $0.1$ \\
 \hline
\end{tabular}
\caption{\label{tab:interactions}Specific parameters (in units of the nearest-neighbor hopping $t_z$ on the $\gamma$ band) chosen of the interaction term in Eq. \eqref{eq:interaction-range} corresponding to the dot in Fig. \ref{fig:phase-diagram}. Note that in Fig. \ref{fig:phase-diagram}, only $v_{x,x}(2),v_{x,y}(2)$ were varied. }
\end{table}

Fig. \ref{fig:phase-diagram} shows a cut through the $T=0$ phase diagram in response to intra- and inter-band couplings, obtained by solving the self-consistent BCS equations described in Eq. \eqref{eq:nonlinear-gap}. 
Indeed, as previously argued, by Table \ref{tab:symmetry} and \ref{tab:harmonics}, a $g$ wave SC gap with range $R\le 2$ must satisfy $\Delta_z = \Delta_x +\Delta_y=0$ and thus must be constructed via the anticipated form $\Delta_{x/y}(\k) =\pm \D'(\k)$.
Moreover, since $\hd'_R(\k)$ is first nonzero for $2^\text{nd}$ neighbor $R=2$, the range $R=2$ interactions (including both intra- $v_{x,x},v_{z,z}$ and inter-band $v_{x,y}$) are responsible for generating $d_{xy}$- or $g$-wave symmetries.
The remaining interaction parameters are fixed and given in Table \ref{tab:interactions}, tuned so that at $T=0$ there exists a background $d_{x^2-y^2}$-wave.
Throughout, the preferred phase relation between the two components is such as to break time-reversal symmetry, i.e. $d_{x^2-y^2} \pm i g$ or $d_{x^2-y^2} \pm i d_{xy}$.

In the vast parameter space of interactions, it is probably unsurprising that we can find regions in which the $d_{x^2-y^2}$- and $g$-waves coexist.  However, while in any single band problem, such coexistence regions are generically  exceedingly narrow, here the coexistence region is relatively large reflecting the fact that the two components live primarily on different bands, and so hardly compete with one another
\footnote{Along the $x$-axis of Fig. \ref{fig:phase-diagram}, the presence of a $d_{xy}$ wave would imply comparable magnitude on all three bands and thus induces strong competition with the $d_{x^2-y^2}$ wave. 
Therefore, the accidental degeneracy between the $d_{x^2-y^2}$ and $d_{xy}$ state is quickly destroyed by the preference for a pure $d_{xy}$ wave as shown in inset of Fig. \ref{fig:phase-diagram}. 
In contrast, the $g$ wave component mainly lives on the $\alpha,\beta$ bands by symmetry requirements and thus there is a natural ground state where the background $d_{x^2-y^2}$ shifts progressively to the $\gamma$ band to minimize competition.  A similar logic would apply in the instance where the background state is an $s$-wave instead of a $d_{x^2-y^2}$-wave.
}.
Within the coexistence regions, the transition temperatures of the two orders are comparable, i.e., $T_c \sim T_\text{trsb}$.
Indeed, in the inset graph in Fig. \ref{fig:phase-diagram}, we plot the onset temperature $T_\text{onset}$ of each component ($d_{x^2-y^2},g,d_{xy}$ waves) along the given cut (dashed purple line) in the phase diagram.
As the secondary order component transitions from $d_{xy} \to g$ wave (left to right), the corresponding background $d_{x^2-y^2}$ wave has an increase in $T_c$.

\subsection{Splitting of $T_c,T_\text{trsb}$ with uniaxial stress}
\begin{figure}
\centering
\includegraphics[width=1\columnwidth]{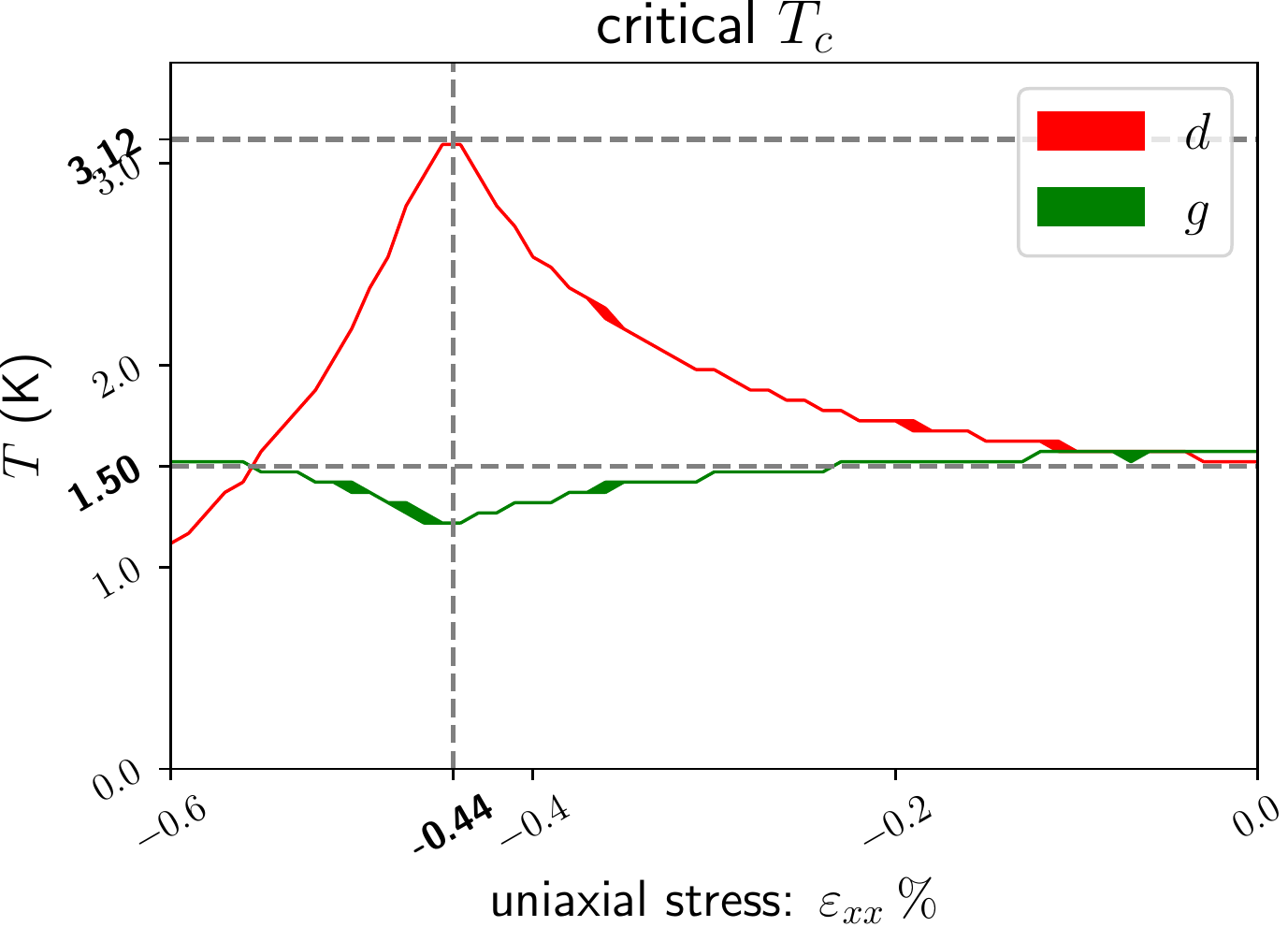}
\caption{Critical Temperatures. The red and green lines represent the critical temperatures of the $d$- and $g$-wave components.
(The thickness denotes numerical uncertainty in the result based on calculations on a $3000^2$-site square lattice.) 
The Lifshitz transition occurs at a uniaxial stress $\ve_{xx} \approx -0.44\%$.}
\label{fig:TK-uniaxial}
\end{figure}

Let us tune the parameters to the dot in Fig. \ref{fig:phase-diagram} so that the ground-state has $d_{x^2-y^2}\pm ig$ pairing and $T_c\approx T_\text{trsb}$ in the absence of strain.
We model the effect of uniaxial stress (parametrized by $\ve_{xx}$) by varying the band parameters in a manner consistent with experiment \cite{li2022elastocaloric,sunko2019direct} such that the $\gamma$-band crosses the van Hove point $(0,\pm \pi)$ at $\ve_{xx} \approx -0.44\%$ (details in Appendix \eqref{sec:strain-details}). 
The $\alpha,\beta$-bands are also distorted, but only slightly \cite{sunko2019direct}.

The resulting critical temperature for the $d$ and $g$-components are given in Fig. \ref{fig:TK-uniaxial}.
Due to symmetry breaking uniaxial stress, the point group symmetry $D_{4h}$ is broken, while the subgroup $D_{2h}$ is preserved.
This implies that the $d_{x^2-y^2}$ component in general gains an $s$-wave component.
Similarly, the $g$-wave should also gain a $d_{xy}$-wave component.
However, our numerical calculations show that the additional $d_{xy}$ is much smaller than the $g$-wave component.
This is presumably due to the fact that the $\alpha,\beta$-bands are minimally distorted in the presence of strain, while the strong $d$ or $d+s$-wave component on the $\gamma$-band suppresses any other symmetry component.
Therefore, we label the corresponding symmetries as given in the legend of Fig. \ref{fig:TK-uniaxial}.

In the presence of small uniaxial stress $\ve_{xx} \gtrsim -0.2\%$, away from the Lifshitz transition, the critical temperatures remain remarkably close to each other, and thus the accidentally degenerate $d+ig$-state is stable when the DOS remains roughly constant.
Near the van Hove point $\ve_{xx}\approx-0.44 \%$, the $d$-wave channel is enhanced causing the observed split in $T_c >T_\text{trsb}$.
At a value of the strain somewhat beyond the van Hove point, $\ve_{xx} \lesssim -0.5\%$, our calculations suggest that the $g$-wave component is first observed when entering superconductivity.
However, note that in the context of SRO, 
there exists an observed competing spin-density-wave (SDW) \cite{grinenko2021split} which is not incorporated in our calculations.
\subsection{Response to pure $B_2$ and $A_1$ strain}

\begin{figure}[ht]
\subfloat[\label{fig:B2}]{%
  \centering
  \includegraphics[width=0.8\columnwidth]{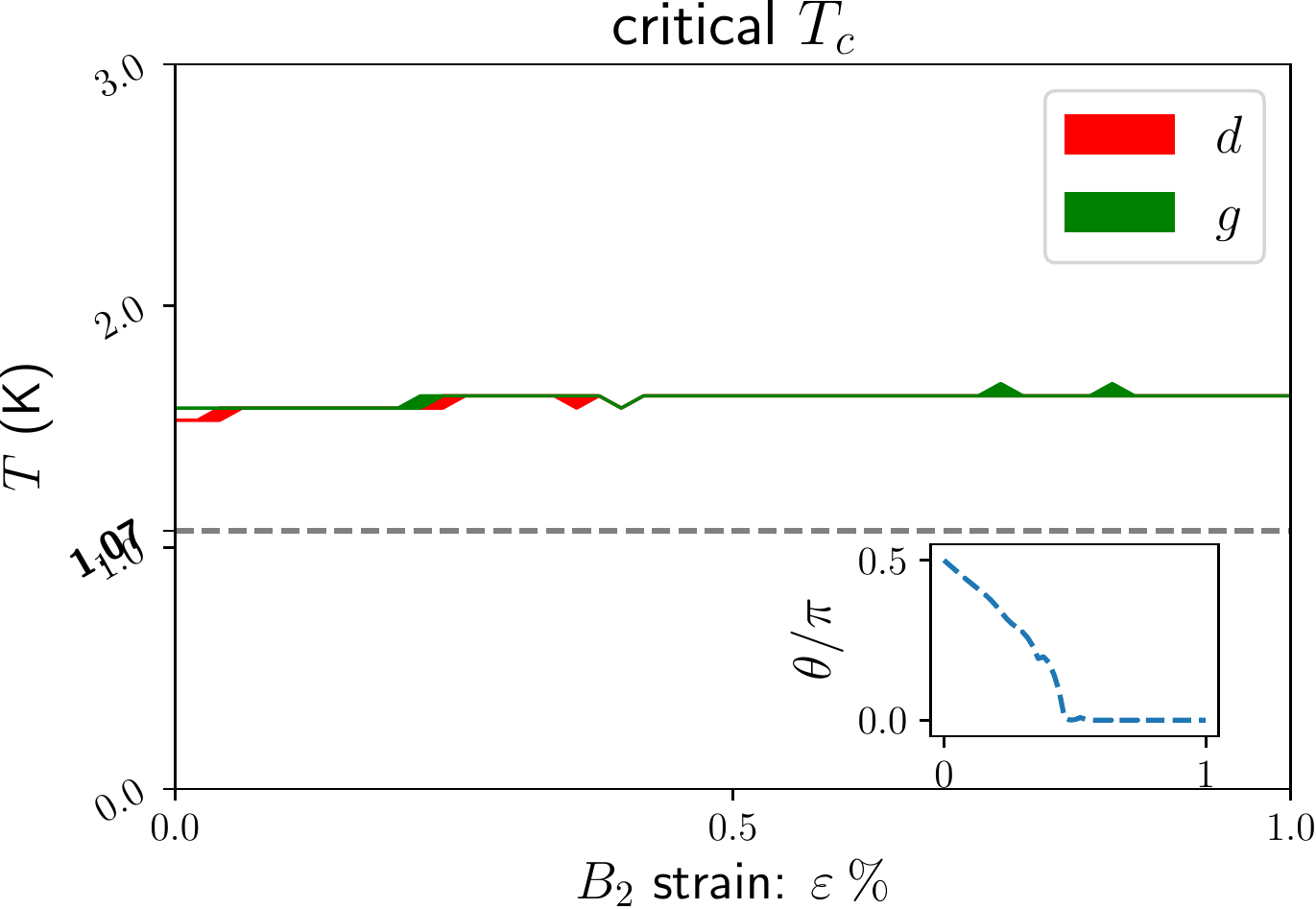}
}
\\
\subfloat[\label{fig:A1}]{%
  \centering
  \includegraphics[width=0.8\columnwidth]{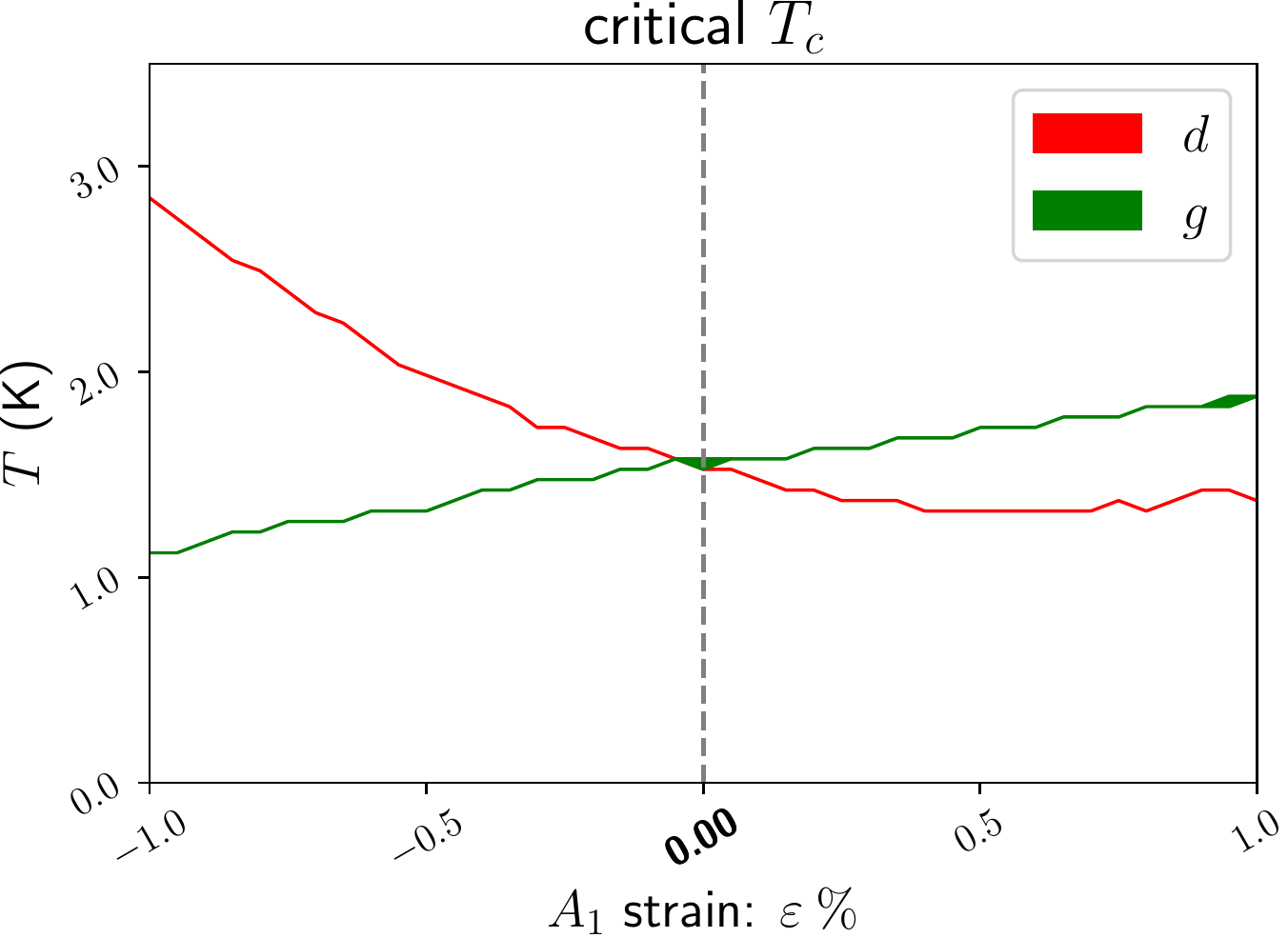}
}
\caption{The red and green lines represent the critical temperature $T$ at which the $d$- and $g$-waves occur under (a) pure $B_2$ strain and (b) pure $A_1$ strain. 
The inset in subplot (a) shows the relative phase between the $d$ and $g$-waves near $T_c$, i.e., along the horizontal dashed line in subplot (a).}
\label{fig:pure-strain}
\end{figure}

We again tune the parameters to the dot in Fig. \ref{fig:phase-diagram} so that the $d,g$ waves are accidentally degenerate and investigate the response of our multi-band model to pure $B_2$ and pure $A_1$ strain. 
The resulting critical temperatures regarding the onset of $d$ and $g$ waves are shown in Fig. \ref{fig:pure-strain}.
In Fig. \ref{fig:B2}, symmetry breaking $B_2$ strain is simulated by modifying the 2$^\text{nd}$ nearest neighbor ($R=2$) hoppings in the normal state band structure (details in Appendix \eqref{sec:strain-details}).
In this case, the critical temperatures remain remarkably close to each other even at large strain values (at which the relative phase between the $d,g$ components is $\theta=0$) and thus provides further evidence to the stability of the accidental degeneracy in the presence of strain. 
In Fig. \ref{fig:A1}, we introduced pure $A_1$ strain by tuning the relative energies of the three bands so that the $\gamma$ band is driven towards (away from) the Van Hove singularity for negative (positive) values $\ve$, with the total density of electrons held fixed  (see Appendix \eqref{sec:strain-details}).
This results in a cross over between the $d,g$ waves, i.e., for $\ve <0$ the $d$-wave is favored over the $g$ wave and vice-versa, resulting in a kink-like feature in the overall $T_c$ (maximum of the two lines in Fig. \ref{fig:A1}) at $\ve=0$.  
No such kink has been detected in experiments on SRO to date \cite{jerzembeck2022superconductivity}.
\section{Summary}

Although \sro{} would seem to be the ideal material to serve as the model system for unconventional superconductivity, given that its normal state is an extremely well characterized Fermi liquid, even the symmetries of the superconducting state has remained controversial \cite{mackenzie2017even}. 
Ultimately, this issue can only be settled by either reproducible phase-sensitive measurements, or by direct imaging (e.g. by angle-resolved photoemission or quasi-particle interference) of the gap structure on the Fermi surfaces. 
However, in the absence of these, further progress must rely on more indirect evidence based on comparisons between relatively robust aspects of  theory and experimentally detected trends. 
Some of these aspects depend only on symmetry.  However, while  microscopic aspects of the problem are more difficult to access unambiguously, certain features, especially those that relate directly to qualitative aspects of the known band-structure, can be useful for the present purposes.

Here, we have analyzed a simple microscopic model with the band-structure of \sro{} and effective pairing interactions treated as phenomenological input 
(we make no claims to the origin of the effective interaction).
We have found several features of the solution of this problem that lend credence to the conjecture that the peculiar difficulty in settling the order parameter question arises from an accidental near-degeneracy between a $d_{x^2-y^2}$-wave and $g$-wave pairing tendency. Specifically, we find the following suggestive results:
\begin{itemize}
    \item It is sufficient to consider a model with relatively short-ranged pairing effective interactions - out to second neighbor distance.
    In common with many previous studies \cite{scaffidi2014pairing,agterberg1997orbital,raghu2010hidden,huo2013spin}, we find that the generic result is that pairing is dominant either on the $\gamma$ band or on the $\alpha$ and $\beta$ bands.  Thus, in any case, a certain degree of fine-tuning of the interactions is necessary to insure that the gap magnitude is comparable on all bands (as is experimentally  established \cite{lederer}).
    \item  Under circumstances of near degeneracy, the $d_{x^2-y^2}$-wave pairing occurs dominantly on the $\gamma$ band and the $g$-wave on the $\alpha$ and $\beta$ bands. Loosely, this near-degeneracy requires no more fine-tuning than is required to have comparable gaps on all bands.
    \item As a consequence, if the band-structure is tuned (for instance by uniaxial stress) such that the $\gamma$ band approaches the nearby Van Hove point, this can significantly enhance the $d$-wave component of the order parameter, but has relatively little effect on the $g$-wave.
    (Generally, if two components coexist in the same bands, then whatever enhances one  tends to suppress the other.)
    \item The existence of gapless (nodal) quasi-particles is protected by symmetry, even if time-reversal symmetry is broken in a dominantly singlet $d_{x^2-y^2}+ig$ state, so long as the mirror symmetry $\sM_z$ (along with the usual $\sM_{xy}$) of the crystal is unbroken. Conversely, similar to the single-band model, $\sM_{xy}$ breaking strain can induce a SC gap which can be detected in experiment.
\end{itemize}

While none of these results is sufficiently unique to serve as confirmation of the basic scenario, they serve to increase confidence in its ``naturalness.'' 

\textit{Note added}. During the preparation of the current manuscript, the following pre-print \cite{PhysRevB.106.134512} appeared on arXiv, in which the authors stabilized a $d_{x^2-y^2} +ig$ state via introducing a 2nd nearest-neighbor term in the bare interaction. Although the calculation was done in the weak coupling limit in terms of RPA, it agrees with our proposal that a relatively short-ranged model is sufficient to stabilize a $g$-wave component. 
In contrast, the following paper \cite{sheng2022multipole} proposes a multipole-fluctuation pairing mechanism to stabilize a $d_{x^2-y^2} +ig$ state. In doing so, they assumed a long-range interaction ($\sim 1/q^2$ in momentum space) which may have been the reason for stabilizing a secondary $g$-wave symmetry.

\begin{acknowledgments}
SAK and AY were supported, in part,  by NSF grant No.  DMR-2000987 at Stanford. EB was supported by the European Research Council (ERC) under grant HQMAT (Grant Agreement No. 817799), the Israel-US Binational Science Foundation (BSF), and a Research grant from Irving and Cherna Moskowitz.
\end{acknowledgments}

\bibliography{main.bbl}

\begin{thebibliography}{72}
\bibitem [{slo()}]{slope}%
  \BibitemOpen
  \href@noop {} {}\bibinfo {note} {{They both generate either $d_{x,y}$- or
  $g$-wave symmetry. The intuition (numerically justified) is that a
  sufficiently repulsive (positive) $v_{x,y}$ would imply that the SC order
  parameter changes sign between the $x,y$ bands and thus preferring a
  $g$-wave. The positive slope implies a repulsive (positive) $v_{z,z}$ and
  thus discourages the $d_{xy}$ wave on the $z$-band. Conversely, attractive
  (negative) $v_{x,y},v_{z,z}$ would generate a $d_{xy}$ wave on the
  $x,y$-bands and $z$-band, respectively.}}\BibitemShut {Stop}%
\end{thebibliography}
\onecolumngrid
\section{Appendix}
\subsection{Band structure}
\label{sec:band-structure}
The normal state dispersion values used in this paper are found in the Supplementary Information (SI) of Ref. \cite{burganov2016strain}. 
We repeat them here for completeness, i.e.,
\begin{align}
    \label{eq:band-structure-details}
    \ve_{x/y} (\k) &= -\mu_0 - 2t_x \cos k_{x/y} -2t_y \cos k_{y/x} \\
    \ve_z (\k) &= -\mu_z -2t_z (\cos k_x +\cos k_y) -2t_z'(2 \cos k_x \cos k_y) \\
    \hy (\k) &= 4 t_h \sin k_x \sin k_y
\end{align}
Where for \sro{}, we use the values
\begin{table}[h!]
\centering
\begin{tabular}{|r|r|r|r| r| r| r||}
 \hline
 $t_z$ (meV)  & $t_x$ (meV)   & $t_y/t_x$   &  $t_h/t_x$ & $\mu_z/t_z$ & $\mu_x/t_x$ \\ [0.5ex]
 \hline\hline
 119  & 165   & 0.08    & 0.13 &  1.48 & 1.08 \\\hline
\end{tabular}
\end{table}
The SO coupling was taken to be $\soc{}=30$ meV.
\subsection{Symmetry Operators}
\label{sec:symmetry-operations}
We shall briefly discuss the explicit forms and properties of symmetry operators, i.e., particle-hole symmetry (PHS) $\sC$, time-reversal symmetry (TRS) $\sT$, and generators of the point group $D_{4h}$, i.e., parity $\sP$, mirror symmetries $\sM_z,\sM_{xy}$ and $\pi/2$-rotation symmetry $\sC_4$. Notice that, except for PHS $\sC$, all other symmetries discussed will be restricted to the particle space. Their extension to particle-hole space follows uniquely that they commute with PHS $\sC$. More specifically, $\bra a_\text{h}|\sO|b_\text{h}\ket = \bra a_\text{p}|\sO^*|b_\text{p}\ket$ where $\sO^*$ denotes complex conjugation of symmetry operator $\sO$ and the subscript $\text{h},\text{p}$ denote hole, particle states..

\begin{enumerate}
    \item PHS $\sC$ is an anti-unitary map on particle-hole space, which take particle state to hole state and vice-versa, and maps the wave vector $\k\mapsto -\k$. In particular, $C^2=1$.
    \item TRS $\sT$ is regarded as anti-unitary map on particle space, which acts as $\exp(-i\pi \sigma_2/2)$ on spin space, acts trivially on orbital space and maps the wave vector $\k\mapsto -\k$. In particular, TRS reverse the spin operators, i.e., $\sT \sigma_i \sT^\dagger = -\sigma_i$ where $\sigma_i$ are the Pauli matrices acting on spin space.
    \item Parity $\sP$ is consistent with its real space transform $x,y,z\mapsto -x,-y,z$. Hence, it acts trivially on spin and orbital space, and maps $\k\mapsto -\k$.
    \item Mirror symmetry $\sM_z$ is consistent with its real space transform $x,y,z\mapsto x,y,-z$. More specifically, it acts as $\exp(-i\pi \sigma_3/2)$ on spin space (an $SU(2)$ action which canonically maps to $\sP \sM_z \in SO(3)$),  maps $d_{yz},d_{zx},d_{xy}\mapsto -d_{yz},-d_{zx},+d_{xy}$ orbitals, and maps the wave vector $\k \mapsto \sM_z^{-1} \k \equiv (k_x,k_y,-k_z)$. Due to spin space, $\sM_z^2 =-1$
    \item Mirror symmetry $\sM_{xy}$ is consistent with its real space transform $x,y,z\mapsto y,x,z$. More specifically, it acts as $\exp(-i\pi (\hat{n} \cdot \sigma)/2)$ on spin space where $\hat{n}$ is the unit vector in $(-1,1,0)$ direction (an $SU(2)$ action which canonically maps to $\sP \sM_{xy} \in SO(3)$),  maps $d_{yz},d_{zx},d_{xy}\mapsto d_{zx},d_{yz},d_{xy}$ orbitals, and maps the wave vector $\k \mapsto \sM_{xy}^{-1} \k \equiv (k_y,k_x,k_z)$. Due to spin space, $\sM_{xy}^2 =-1$
     \item $\pi/2$-rotation $\sC_4$ is consistent with its real space the transform $x,y,z\mapsto y,-x,z$. More specifically, it acts as $\exp(-i\pi \sigma_3 /4)$ on spin space (an $SU(2)$ action which canonically maps to $\sC_4 \in SO(3)$),  maps $d_{yz},d_{zx},d_{xy}\mapsto -d_{zx},d_{yz},-d_{xy}$ orbitals, and maps the wave vector $\k \mapsto \sC_4^{-1} \k \equiv (-k_y,k_x,k_z)$.
\end{enumerate}

Notice that within certain planes (e.g., $k_z=0,\pi$ plane for $\sM_z$ and $k_x=k_y$ plane for $\sM_{xy}$), the mirror symmetries $\sM_z,\sM_{xy}$ are local in $\k$ and thus only act non-trivially in the orbital/spin space.
Therefore, it's useful to define $\lz,\lxy$ as their orbital/spin actions so that $\sM_z =\lz \otimes (k_z\mapsto -k_z)$ and $\sM_{xy} =\lxy \otimes (k_x\leftrightarrow k_y)$
\subsubsection{Normal State Symmetries}
\label{sec:normal-state}
Using the explicit forms of point group symemtries of $D_{4h}$ discussed in the previous section, one finds that the matrix components of the normal state dispersion are mapped in a rather complicated manner and thus tabulated in Table \ref{tab:normal-state-symmetry-mapping}.
Since the normal state is invariant under the point group $D_{4h}$, it's then clear that  $\ve_x +\ve_y, \ve_x-\ve_y, \ve_h$ satisfy $S,D,D'$-wave ($A_{1g},B_{1g},B_{2g}$) symmetry in $\k$-space, respectively.
Compare this to the specific parametrization in Appendix \eqref{sec:band-structure}. A similar analysis can be done for the SC state with regards to the SC gap $\Delta(\k)$.

\begin{table}[h!]
\centering
\begin{tabular}{||c |c||}
 \hline
 Operation $g$   & Mapping \\ [0.5ex]
 \hline\hline
 $\sC_4$   & 
 $\begin{aligned}
      \ve_z (\k) &\mapsto \ve_z(\k') \\
     (\ve_x \pm \ve_y) (\k) &\mapsto \pm (\ve_x \pm \ve_y) (\k')\\
      \hy (\k) &\mapsto -\hy (\k')\\
      \soc{} (\k) &\mapsto \soc{} (\k')
 \end{aligned}$ \\
 \hline
 $\sM_{x,y}$   & 
 $\begin{aligned}
     \ve_z (\k) &\mapsto \ve_z(\k')\\
     (\ve_x \pm \ve_y) (\k) &\mapsto \pm (\ve_x \pm \ve_y) (\k')\\
      \hy (\k) &\mapsto +\hy (\k')\\
      \soc{} (\k) &\mapsto \soc{} (\k')
 \end{aligned}$\\
 \hline
\end{tabular}
\caption{\label{tab:normal-state-symmetry-mapping} 
Mapping of group operations $\sO = \sC_4, \sM_{xy}$ of $D_{4h}$ where $\k' \equiv \sO^{-1} \k$.}
\end{table}

\subsection{Existence of Line Node}
\label{sec:nodal-point}

Continuing the argument from Section \eqref{sec:existence}, notice that $H_+(\hat{\ve},|\Delta_3|)$ can be decomposed into $= \tau_1 \lambda_3 |\Delta_3| + \tau_3\hat{\ve}_+$ (where $\lambda_3=\text{diag}(1,-1,0) $ is the 3rd Gell-Mann matrix) and thus anti-commutes with $\tau_2$.
Hence, the determinant (which is invariant under unitary transforms) is given by
\begin{equation}
    \det H_+(\k) = |\det (\hat{\ve}_+(\k) + i|\Delta_3(\k)| \lambda_3)|^2
\end{equation}
Where $\k$ is restricted along the $(110),(1\bar{1}0)$ directions.
Therefore, there exists a nodal point if and only if $\det (\hat{\ve}_+ + i|\zeta| \lambda_3) =0$. 
Notice that mirror symmetry $\sM_{xy}$ maps $\lambda_3\mapsto -\lambda_3$ while keeping $\hat{\ve}_+$ invariant. Therefore, $\det (\hat{\ve}_+ + i|\Delta_3| \lambda_3) =\det (\hat{\ve}_+ -i|\Delta_3| \lambda_3)$ is a real value along the $(110),(1\bar{1}0)$ directions. 
The nodal equation $\det (\hat{\ve}_+ + i|\Delta_3| \lambda_3) =0$ thus corresponds to one real equation (instead of 2 where both the real and imaginary components are required to be zero), and thus defines a line node along the $(110),(1\bar{1}0)$ directions. 

For completeness, the nodal equation can be computed explicitly to be
\begin{equation}
  \label{eq:nodal-withSOC}
  |\Delta_3(\k)|^2 \ve_z(\k) + \det \hat{\ve}(\k) = 0
\end{equation}
And in the absence of SOC, $\soc=0$, the equation reduces to
\begin{equation}
  \label{eq:nodal-withoutSOC}
  |\Delta_3(\k)|^2 = \hy(\k)^2 -\ve_x(\k)^2
\end{equation}
Since the gap function is usually the smallest scale in the system, the equation generally holds true at some point along the diagonals, near the intersection of the FS of the $x,y$ bands, despite the fact that $\Delta_3 \ne 0$ \footnote{If on the unlikely event that the $|\Delta_3|$ is sufficiently large so that the Eq. \eqref{eq:nodal-withoutSOC} never holds, there still exists a nodal point along the $\gamma$ band since SOC is assumed to be zero for \eqref{eq:nodal-withoutSOC}.}.
\subsection{Rigorous perturbative argument}
\label{sec:stability-details}

\subsubsection{Setup}
Let us consider a general 3D multi-band $d_{x^2-y^2}+ig$ state $H(\k)$ with spin degree of freedom so that $H(\k)$ is a $4n\times 4n$-matrix at each wave vector $\k$ ($n=3$ is the number of orbitals in \sro{}). We shall assume that $H(\k)$ has arbitarily strong SOC and dominantly singlet with possible subdominant triplet pairing and weak $k_z$-dependence so that we can treat the triplet pairing and $k_z$ dependence perturbatively. 
Let $H^\text{s}(\k)$ be obtained from $H(\k)$ by setting the triplet pairing gap elements $=0$, i.e., the purely singlet portion so that in particle-hole space,
\begin{equation}
    H(\k) = H^\text{s}(\k) + 
    \begin{bmatrix}
      0 & \Delta^\text{triplet}(\k) \\
      \Delta^\text{triplet}(\k)^\dagger & 0
    \end{bmatrix}
\end{equation}

Let us further define
\begin{align}
    H_0^\text{s}(\k) &= \sum_{m=\pm i} 1_{\lz = m} H^\text{s}(\k) 1_{\lz = m} \\
    H_0(\k)&= \sum_{m=\pm i} 1_{\lz = m} H(\k) 1_{\lz = m}
\end{align}
Where $1_{\lz = \pm i}$ are the projection operators onto the eigenspaces of the orbital/spin action $\lz = \pm i$. In this case, $H_0(\k), H_0^\text{s}$ are decoupled in the eigenspaces $\lz=\pm i$ which was the essential ingredient in our restricted 2D model in the main text. 
It is worth mentioning that here, $H_0(\k)$ and $H_0^\text{s}(\k)$ may have $k_z$-dependence.

Notice that mirror symmetry $\sM_{xy}$ maps $\Delta(\k) \mapsto -\Delta(\k)$ so that after a gauge transform $\sU = i\tau_3$ (where $\tau_3$ is the pauli matrix in particle-hole space), the Hamiltonian $H(\k)$ is invariant $\sU \sM_{xy}\equiv \tilde{\sM}_{xy}$. 
Subsequently, the restricted Hamiltonians $H^\text{s},H_0,H_0^\text{s}$ are also invariant under $\tilde{\sM}_{xy}$.
\subsubsection{Starting Point: Decoupled Singlet Hamiltonian}
As discussed in the main text, since $H_0^\text{s}(\k)$ is of purely singlet pairing and is decoupled in the eigenspaces of $\lz=\pm i$, it has a line node extended in the $k_z$ direction in the $(110)$ direction, which we denote by $\k_0^\text{s}(k_z)$.
Let us fix $k_z$, and thus from Theorem \eqref{theorem:quartic}, it is clear that $\dim \ker H_0^\text{s}(\k_0^\text{s})=4m$ for some $m\ge 1$.
In the absence of other symmetries, it is generally assumed that $m=1$ so that we can choose $\phi_0^\text{s}(\k_0^\text{s})$ so that $\{\phi_0^\text{s}(\k_0^\text{s})\}=(\phi_0^\text{s}, \lxyg\phi_0^\text{s}, \sC\sP \phi_0^\text{s}, \sC\sP \lxyg\phi_0^\text{s})$ forms a basis for $\ker H_0^\text{s}(\k_0^\text{s})$, and we shall use $P_0^\text{s}(\k_0^\text{s})$ to denote the projection operator onto $\ker H_0^\text{s}(\k_0^\text{s})$.

\subsubsection{Adding mirror symmetric perturbations, e.g., weak triplet pairing}
Let us fix $k_z$ as before, and let $\k$ be near $\k_0^\text{s}(k_z)$ along the $(110)$ direction so that $\tilde \sM_{xy} = \lxyg$.

Since triplet pairing is subdominant, we can assume that $H_0(\k) -H_0^\text{s}(\k_0^\text{s})$ with respect to the gap of $H_0^\text{s}(\k_0^\text{s})$ (lowest nonzero energy level of $H_0^\text{s}(\k_0^\text{s})$ and not to be confused with the gap function $\Delta(\k_0^\text{s})$) when $\k$ is near $\k_0^\text{s}$ so that $P_0(\k)$ can be well-defined as the projection operator onto eigenstates of $H_0(\k)$ with energies near zero and also unitarily equivalent to $P_0^\text{s}(\k_0^\text{s})$ as explained in subsection \eqref{sec:Riesz}.
By perturbation theory (Weyl's inequality), the variation of the energy values must be small and thus it's sufficient to consider the restriction of $H_0(\k)$ to $P_0(\k)$ when attempting find the zero energies of $H_0(\k)$, i.e., only eigenstates in $P_0(\k)$ have energies sufficiently close to zero.

Since $H_0(\k)$ is decoupled in the eigenspaces of $\lz=\pm i$ and invariant under mirror symmetry $\lxyg$, from Theorem \eqref{theorem:quartic}, there must exist $\phi_0(\k)$ such that $\{\phi_0(\k)\}$ forms a complete orthonormal basis for $P_0(\k)$ (though not necessarily degenerate anymore).
As an analogy, $\phi_0(\k),\lxyg\phi_0(\k)$ should be considered as \textbf{``pseudo-spins"} in particle space, while the remaining states are considered as ``pseudo-spins" in hole states.
Within the basis $\{\phi_0(\k)\}$, the restriction $H_0(\k)P_0(\k)$ must be of the form
\begin{equation}
    H_0(\k)P_0(\k) = 
    \begin{bmatrix}
      \hat{\ve}_0(\k) & \hat{\Delta}_0(\k) \\
      \hat{\Delta}_0(\k)^\dagger & -\hat{\ve}_0(-\k)^T
    \end{bmatrix}
\end{equation}
By $\sC\sP$ symmetry, $\hat{\Delta}_0(\k)$ must form a singlet between the pseudo-spins, i.e., in the basis representation $\{\phi_0(\k)\}$, the matrix $\hat{\Delta}_0(\k)=\Delta_0(\k) (i\sigma_2)$ where $\Delta_0(\k)\in \mathbb{C}$ and $\sigma$ denotes the Pauli matrix in pseduo-spin basis.
Furthermore, under the basis $\{\phi_0(\k)\}$, the symmetries $\lz,\lxyg$ have the matrix forms
\begin{equation}
    \label{eq:pseudospin-matrix}
    \lz = 
    \begin{bmatrix}
      +i &&& \\
      & -i && \\
      && -i & \\
      &&& +i
    \end{bmatrix}, \quad 
    \lxyg = 
    \begin{bmatrix}
       & +1 && \\
      +1 &  && \\
      &&  & +1 \\
      && +1 & 
    \end{bmatrix}
\end{equation}
Since $H_0(\k)P_0(\k)$ preserves the local symmetries $\lz,\lxyg$, explicit calculations show that $\hat{\ve}_0(\k) = \ve_0(\k) \sigma_0$ where $\ve_0(\k)\in \mathbb{R}$ and $\hat{\Delta}_0(\k)=0$. 
Therefore,
\begin{equation}
    H_0(\k)P_0(\k) = \tau_3 \otimes \ve_0(\k) \sigma_0
\end{equation}
Therefore, along the $(110)$ direction, the triplet-mixing Hamiltonian $H_0(\k)$ has a nodal point at $\k$ if and only if $\ve_0(\k) = 0$. 
Since this a real equation on a 1D line, it's in general satisfied at a single point node $\k_0 (k_z)$ near $\k^\text{s}(k_z)$ at each $k_z$ and thus extended line nodes are stable under triplet mixing as long as the symmetries $\lz,\lxyg$ are satisfied.

\subsubsection{Adding weak coupling between eigenspaces of $\lz =\pm i$}

Let us fix $k_z$ as before, and let $\k$ be near $\k_0(k_z)$, but not necessarily in the $(110)$ direction.

Since \sro{} is quasi-2D, the coupling between eigenspaces of $\lz =\pm i$ must be weak and thus we can assume that $H(\k) -H_0(\k_0)$ with respect to the gap of $H_0(\k_0)$ (lowest nonzero energy level of $H_0(\k_0)$ and not to be confused with the gap function $\Delta(\k_0)$) when $\k$ is near $\k_0$ so that $P(\k)$ can be well-defined as the projection operator onto eigenstates of $H(\k)$ with energies near zero and also unitarily equivalent to $P_0(\k_0)$ as explained in subsection \eqref{sec:Riesz}.
By perturbation theory (Weyl's inequality), the variation of the energy values must be small and thus it's sufficient to consider the restriction of $H_0(\k)$ to $P_0(\k)$ when attempting find the zero energies of $H_0(\k)$, i.e., only eigenstates in $P_0(\k)$ have energies sufficiently close to zero.

\textbf{I. First assume} that $\k$ is still along the $(110)$ direction so that $\lxyg$ is a local symmetry.
It should be noted from the matrix form of $\lxyg$ with respect to $\{\phi_0(\k_0)\}$ in Eq. \eqref{eq:pseudospin-matrix} that \begin{equation}
    \tr (\lxyg P_0(\k_0)) = 0
\end{equation}
Since $\lxyg$ has eigenvalue $\pm 1$, we see that $\tr (\lxyg P_0(\k_0))$ must be an integer and thus by continuity, for sufficiently small variations $H(\k)-H_0(\k_0)$, we must also have
\begin{equation}
    \tr (\lxyg P(\k)) = 0
\end{equation}
Since $H_0(\k)$ still preserves mirror symmetry $\lxyg$ in the diagonal planes and $\dim P(\k)=\dim P_0(\k_0)=4$, we see that eigenstates of $H(\k)$ in $P(\k)$ can be divided by the eigenspaces of $\lxyg=\pm 1$, each of which is of dim 2, i.e., $\dim \ker (\lxyg\pm 1) =2$.
Notice that $\sC\sP$ commutes with $\lxyg$ and thus maps within each eigenspace $\lxyg=\pm 1$.
Hence, by Lemma \eqref{lemma:double}, there exists an orthonormal basis $\chi_\pm, \sC\sP \chi_\pm$ consisting of eigenstates of $H(\k)$ within each eigenspace $\lxyg=\pm 1$, respectively.
Let us now define
\begin{equation}
    \phi(\k) = \frac{1}{\sqrt{2}} (\chi_+ +\chi_-)
\end{equation}
Then it's clear that $\{\phi(\k)\}$ forms a complete orthonormal basis for $P(\k)$. 
We can then repeat the same argument as before, restricting to the effective 2-band Hamiltonian $H(\k)P(\k)$, but only using mirror symmetry $\lxyg$ this time.
In this case, we see that $\hat{\Delta}(\k) =0$ while $\hat{\ve}(\k) = \ve(\k) \sigma_0 +\ve'(\k) \sigma_1$, or equivalently,
\begin{equation}
    H(\k)P(\k) = \tau_3 \otimes (\ve(\k) \sigma_0 +\ve'(\k) \sigma_1)
\end{equation}
Therefore, the perturbed Hamiltonian $H(\k)$ has a node if and only if $|\ve(\k)|=|\ve'(\k)|$. 
Since this we can either take  $\ve(\k)=+\ve'(\k)$ or $\ve(\k)=-\ve'(\k)$, the originally single node $\k_0$ for $H_0$ splits into 2 nodes along the $(110)$ direction for $H(\k)$ and thus suggests that there exists a Bogliubov Fermi surface circling the original node $\k_0$ and intersecting the diagonal plane at the 2 nodes.

\textbf{II. To check this, let us assume the general case} where $\k$ is not necessarily along the $(110)$ direction, but still near the original node $\k_0$.

A similar argument can be repeated to show that there exists a complete basis $\{\phi(\k),\chi'(\k),\sC\sP \phi(\k),\sC\sP \chi(\k)\}$ for $P(\k)$, where $\phi(\k),\chi(\k)$ are not related by mirror symmetry $\lxyg$ anymore.
However, since $\sC\sP$ is still preserved, the effective 2-band gap function $\hat{\Delta}(\k)$ must form a singlet between the pseudospins, i.e., $\hat{\Delta}(\k) =\Delta(\k)(i\sigma_2)$.
This implies that the particle-hole dispersion relation $\hat{\ve}(\k)$ can be unitarily transformed via any $SU(2)$ rotation in pseudospin space without affecting the gap function.
Therefore, the effective 2-band Hamiltonian $H(\k)P(\k)$ can be unitarily transformed so that $\hat{\ve}(\k)$ is diagonal with entries $\ve_\pm (\k)$, i.e.,
\begin{equation}
    H(\k)P(\k)  = 
    \left( \begin{array}{c|c}
      \begin{matrix}
          \ve_{+}(\k) &        \\
           & \ve_{-}(\k)
      \end{matrix} 
      &  
      \begin{matrix}
           & \Delta(\k)     \\
           -\Delta(\k) & 
      \end{matrix}  \\
      \cmidrule[0.4pt]{1-2}
      \begin{matrix}
           & -\Delta(\k)^*      \\
           \Delta(\k)^* & 
      \end{matrix} 
      & 
       \begin{matrix}
          -\ve_{+}(\k) &        \\
           & -\ve_{-}(\k)
      \end{matrix} 
    \end{array} \right)
\end{equation}
Where the solid lines divide the particle/hole subspaces.
Therefore, the effective Hamiltonian $H(\k)P(\k)$ has a zero mode if and only
\begin{equation}
    \label{eq:Bogoliubov-FS-appendix}
    |\Delta(\k)|^2 + \ve_+ (\k) \ve_-(\k) = 0
\end{equation}
In the absence of mirror symmetries and TRS, a small Zeeman splitting occurs in pseudospin space, i.e., $\ve_+(\k) \ne \ve_-(\k)$, and thus a Bogoliubov FS is generated where Eq. \eqref{eq:Bogoliubov-FS-appendix} is satisfied.

\subsubsection{Some useful exact statements regarding symmetries}
\begin{lemma}
    \label{lemma:double}
    Let $A$ be a Hermitian operator on a $2N$-dimensional Hilbert space and $\sT$ be an anti-unitary operator which anti-commutes with $A$ and satisfies $\sT=\sT^\dagger=\sT^{-1}$. Then there exists an orthonormal basis of the form $\phi_i, \sT \phi_i,i=1,2,...,N$ where $\phi_i$ corresponds to eigenvalue $a_i$ and $\sT \phi_i$ corresponds to eigenvalue $-a_i$. In particular, 
    \begin{equation}
        \dim \ker A \in 2\mathbb{N}
    \end{equation}
\end{lemma}

\begin{proof}
    It's clear that if $\phi$ is an eigenstate of eigenvalue $a$, then $\sT \phi$ is an eigenstate of eigenvalue $-a$. 
    Therefore, if $a\ne 0$, then $\phi,\sT \phi$ must be orthogonal, and thus the summation of all eigenspaces with nonzero eigenvalue $a\ne 0$ must be of even dimension.
    Since the total Hilbert space is $2N$-dimensions, we see that $\dim \ker A \in 2\mathbb{N}$.
    
    Let us now consider the restriction $\sT_A$ of $\sT$ to $\ker A$.  Since $\sT_A^2 = I_A$ is the identity operator on $\ker A$, we see that $(\sT_A -I_A)(\sT+I_A) =0$. Hence, if $\phi\in \ker A$, then either $\tilde\phi \equiv (\sT_A+I_A)\phi$ is nonzero and thus an eigenstate of $\sT_A$ with eigenvalue $+1$, OR $\tilde\phi = 0$ and thus $\phi$ is an eigenstate of $\sT_A$ with eigenvalue $-1$. In the latter case, we see that $i\phi$ is an eigenstate of $\sT_A$ with eigenvalue $+1$ since $\sT_A$ is anti-linear. 
    Hence, we can always find an orthonormal basis $\chi_1,...,\chi_{2m}$ for $\ker A$ of eigenstates of $\sT_A$ with eigenvalue $+1$.
    
    Now let us define
    \begin{equation}
        \phi_{2i-1/2i} = \frac{1}{\sqrt{2}} (\chi_{2i-1} \pm i\chi_{2i}), \quad i=1,2,...,m
    \end{equation}
    Then it's clear that $\phi_{2i} = \sT \phi_{2i-1}$ and thus we have found a basis of the form $\phi, \sT\phi$ for $\ker A$ and consequently for the entire Hilbert space.
\end{proof}
\begin{theorem}
\label{theorem:quartic}
If BCS Hamiltonian $h$ commutes with $\lz,\lxyg$, then
\begin{equation}
    \dim \ker h = 4m,  \quad m \ge 0
\end{equation}
More specifically, there exists an orthonormal basis for $\ker h$ of the form 
\begin{equation}
    \{\phi\} \equiv  (\phi_i, \lxyg\phi_i, \sC\sP \phi, \sC\sP \lxyg\phi_i:i=1,...,m)
\end{equation}
\end{theorem}

\begin{proof}
    Since $h$ commutes with $\lz$ and $\lz$ is Hermitian with $\lz^2=-1$, we can decompose the particle-hole subspace into the eigenspaces of $\lz=\pm i$. 
    Notice that $\sC\sP$ commutes with $\lz$ and since the eigenvalue of $\lz$ are imaginary, we see that $\sC\sP$ is an anti-unitary operator between the eigenspaces $\lz = +i\leftrightarrow \lz=-i$.
    In particular, we see that $\dim \ker (\lz \pm i) = 2n$
    
    Since $h$ anti-commutes with $\sC\sP$, we see that if $\phi$ is an eigenstate of $h$ in eigenspace $\lz=+i$ with energy $E$, then $\sC\sP \phi$ is an orthogonal eigenstate of $h$ in  eigenspace $\lz=-i$ with energy $-E$, and vice-versa.
    Similarly, notice that $\lxyg$ anti-commutes with $\lz$ and thus is a unitary operator which maps between the eigenspaces $\lz = +i\leftrightarrow \lz=-i$.
    Therefore, $\sC\sP \lxyg$ is an anti-unitary map which maps each eigenspace onto itself, i.e., $\ker (\lz \pm i)$ is invariant under $\sC\sP \lxyg$.
    Notice that $\sC\sP \lxyg$ is an anti-unitary Hermitian operator which anti-commutes with $h$ within each eigenspace $\ker (\lz \pm i)$.
    Therefore, by Lemma \eqref{lemma:double}, within each eigenspace $\ker (\lz \pm i)$, the kernel of $h$ is doubly degenerate.
    In fact, we can find an orthonormal basis of the form $\phi_i, \sC\sP\lxyg \phi_i,i=1,...,m$ for $\ker h\cap \ker (\sM_z -i)$ (notice that the basis may be empty, i.e., $m=0$, if $\ker h$ is trivial, i.e., only contains the zero eigenstate).
    And by $\sC\sP$ symmetry, we see that the kernel of $h$ is quartic-degenerate.
    More specifically, $\phi_i,\lxyg\phi_i,\sC\sP\phi_i, \sC\sP \lxyg\phi_i$ with $i=1,...,m$ form an orthonormal basis for $\ker h$.
\end{proof}

\subsubsection{The Riesz Projector and perturbation theory}
\label{sec:Riesz}

For this section, let $h_0$ be a Hermitian operator on a finite $N$-dimensional Hilbert space and let $h_\lambda = h_0 +\lambda V$ denote the perturbed Hermitian operator parametrized by the small parameter $\lambda$ and the perturbation is denoted by $V$.
To have a controlled perturbation theory, we require the eigenspaces of $h_0$ to ``evolve smoothly" to those of $h_\lambda$ as we tune the small parameter $\lambda$. 
One way of defining such a smooth evolution rigorously is to apply the Riesz projector \cite{kato2013perturbation,hislop2012introduction}, which we shall define shortly.

For simplicity, let $E^1(h_\lambda)\le E^2(h_\lambda) \le \cdots \le E^N(h_\lambda)$ denote the eigen-energies of $h_\lambda$ in ascending order (where each energy level is repeated by their degeneracy), and let $E_0$ denote a (possibly degenerate) eigen-energy of $h_0$, i.e., there exists some $i$ such that
\begin{equation}
    E^{i-1}(h_0) < E_0 \equiv E^i(h_0) =E^{i+1} (h_0) \cdots =E^{i+m-1} (h_0) < E^{i+m}(h_0)
\end{equation}
It's then clear that we can draw a circle $\Gamma$ in the complex plane $\mathbb{C}$ which only contains $E_0$ in the interior and none of the remaining energy levels of $H_0$.
By the Weyl-inequality, we know that $|E^n(h_\lambda) -E^n(h_0)| \le |\lambda| ||V||$, i.e., the variation in energy levels is controlled by the small parameter $\lambda$. 
Therefore, for sufficiently small $\lambda$, the energies $E^i(h_\lambda),...,E^{i+m-1}(h_\lambda)$ are still contained in the interior of $\Gamma$ while the remaining energy levels of $h_\lambda$ are in the exterior.
Within this neighborhood of $\lambda$ values, we can define the \textbf{Riesz projector} as
\begin{equation}
    P_\lambda \equiv  \frac{1}{2\pi i} \oint_\Gamma \frac{dz}{z-h_\lambda}
\end{equation}
It's then easy to see that $P_\lambda$ is smooth with respect to $\lambda$ and is equal to the projection operator onto the summation of eigenspaces $E^i(h_\lambda),...,E^{i+m-1}(h_\lambda)$.
Since the number of energy levels (repeated by their degeneracy) does not change in the interior of $\Gamma$ for sufficiently small $\lambda$, we see that $\dim P_\lambda$ is constant and $P_\lambda, P_{\lambda'}$ are unitarily equivalent for distinct values of small $\lambda,\lambda'$.

\subsection{Decomposition of the interaction}
\label{sec:interaction}
Consider the interaction $\sV$ in Eq. \eqref{eq:interaction} given by the form $V(\k-\k')$ where we have suppressed the band indices $\mu,\mu'$ for simplicity.
Since $V(\q)$ is assumed to be $D_4$-invariant, we can rewrite it in terms of $s$-wave lattice harmonics with range $R \ge 0$, 
\begin{equation}
    V(\q) = \sum_{R\ge0} v(R) \hs_R(\q)
\end{equation}
Let $|\rr\ket$ denote the Dirac delta function at real space lattice site $\rr$ and $W_R$ denote the $D_4$-invariant subspace spanned by $|\rr\ket$ over all lattice sites $\rr$ of range $R$, i.e., $|\rr|=R$. 
Then in real space, the lattice harmonic $s_R$ can be written as the summation over all lattice sites $\rr$ of range $R$, i.e.
\begin{align}
    \hs_R &= \frac{1}{\sqrt{|W_R|}} \sum_{|\rr|=R} |\rr\ket \\
    \hs_R(\rr) &= \frac{1}{\sqrt{|W_R|}} \delta_{|\rr|=R}
\end{align}
Where $|W_R|$ is the number of lattice sites $\rr$ with range $R$.
Conversely, in momentum $k$-space, the lattice harmonic can be rewritten as
\begin{align}
    s_R(\k-\k')&=  \frac{1}{\sqrt{|W_R|}} \sum_{|\rr|=R} \bra \k -\k'|\rr\ket \\
    &= \frac{1}{\sqrt{|W_R|}} \sum_{|\rr|=R} \bra \k |\rr\ket \bra \rr|\k'\ket \\
    &= \frac{1}{\sqrt{|W_R|}} \bra \k| \left[\sum_{|\rr|=R}|\rr\ket\bra \rr| \right] |\k'\ket\\
    &= \frac{1}{\sqrt{|W_R|}} \bra \k| P_R |\k'\ket
\end{align}
Where $P_R$ denotes the projection operator onto $W_R$. 
Notice that for a given range $R$, the subspace $W_R$ can be decomposed into a unique combination of irreps of the symmetry group $D_4$, and thus $P_R$ must be diagonal in terms of its decomposition, i.e.,
\begin{equation}
    P_R = \frac{1}{\sqrt{W_R}} \sum_{\psi_R} |\psi_R \ket \bra \psi_R| 
\end{equation}
Where the summation is over all possible lattice harmonics $\psi_R$ in $W_R$ corresponding to distinct irreps.
In particular,
\begin{equation}
    s_R(\k-\k') = \bra \k|P_R|\k'\ket =\frac{1}{\sqrt{W_R}} \sum_{\psi_R}\psi_R(\k) \psi_R(\k')^*
\end{equation}

As an example, consider the case $R=1$ in Table \ref{tab:harmonics}. Since only $W_1 = A_1 \oplus B_1 \oplus E$, we see that
\begin{equation}
    \hs_1(\k-\k') = \frac{1}{2} \left( \hs_1(\k)\hs_1(\k') +\hd_1(\k)\hd_1(\k')\right) +\frac{1}{2} \left( p_x(\k)p_x(\k') +p_y(\k) p_y(\k')\right)
\end{equation}
Where $p_x(\k),p_y(\k)$ is an orthonormal basis for the $E$ irrep in $W_1$. 
Notice that the terms projecting into the $E$ irrep are not involved in computing the nonlinear gap equation in Eq. \eqref{eq:nonlinear-gap} since we assumed that the gap function $\Delta$ is of even parity.
Therefore, in the main text, lattice harmonics belonging to the $E$ irrep for any range $R$ are omitted.
\subsection{Diagonalization of 2-band BCS Hamiltonian: $SU(4)\to SO(6)$}
\label{sec:su4-so6}
\subsubsection{Setup}
Ignoring spin-orbit coupling, $\soc(\k) =0$, the 2-band system involving the $\alpha,\beta$-bands decouples from the single $\gamma$ band in the BCS Hamiltonian.
Therefore, we can use the reduced 2-band Nambu spinor (in contrast to the 3-band spinor in Eq. \eqref{eq:BCS-block}),
\begin{equation}
    \label{eq:reduced-Nambu}
    \Psi^\dagger (\k) = 
    \left[\begin{array}{*2c|*2c}
    \psi_{x\up}^\dagger(\k) &  \psi_{y\up}^\dagger(\k) & -\psi_{x\down}(-\k) & -\psi_{y\down}(-\k)
    \end{array}\right]
\end{equation}
To rewrite the decoupled 2-band BCS Hamiltonian.
\begin{equation}
    H(\k) = 
    \left(
    \begin{array}{*2c|*2c}
    \ve_x     & \hy        &  \Delta_x   & \Delta_h     \\
    \hy       & \ve_y      &  \Delta_h   & \Delta_y     \\[0.5ex] \hline
    \Delta_x^\dagger  & \Delta_h^\dagger  &  -\ve_x     & -\hy       \\
    \Delta_h^\dagger  & \Delta_y^\dagger  &  -\hy       &  -\ve_y    \\
    \end{array}\right)
\end{equation}
Where we have kept the $\k$-dependency implicit in the entries since the BCS Hamiltonian $H(\k)$ is local in $\k$, and the solid lines separate the particle-hole subspaces. 
Due to particle-hole symmetry and even parity, the eigen-energies of the local $H(\k)$ is doubly degenerate, i.e., of the form $\pm E$. 
To reduce this degeneracy and make the diagonalization process more transparent, let us consider $H(\k)^2$, i.e.,
\begin{align}
    H(\k)^2 &= 
    \left(
    \begin{array}{*2c|*2c}
    R+h        & a+ib        &  0       & c+id   \\[0.5ex]
    a-ib       & R-h         &  -c-id   & 0      \\[0.5ex] \hline
    0          & -c+id       &  R+h     & a-ib   \\[0.5ex]
    c-id       & 0           &  a+ib    & R-h    \\[0.5ex]
    \end{array}\right) \\
    &= R  +h\sigma_{03} +a\sigma_{01}-b\sigma_{32} -c\sigma_{22} -d\sigma_{12}
     \label{eq:H-square}
\end{align}
Where $\sigma_{\mu\nu} =\sigma_\mu \otimes \sigma_\nu$ is the tensor product of Pauli matrices and
\begin{align}
    R &= \hy^2 +|\Delta_h|^2 +\frac12 \left(\ve_x^2 +|\Delta_x|^2 +\ve_y^2  +|\Delta_y|^2\right) \\
    h &= \frac12 \left( \ve_x^2 +|\Delta_x|^2 -\ve_y^2 -|\Delta_y|^2\right) \\
    a +ib &= \Delta_h^\dagger \Delta_x +\Delta_h \Delta_y^\dagger + \ve_h(\ve_x+\ve_y) \\
    c +id &= \Delta_h(\ve_x-\ve_y) -(\Delta_x-\Delta_y)\hy
\end{align}
Indeed, to solve the nonlinear gap equation, we need to obtain the 1-particle density matrix (1-pdm) $\Gamma(\k)$ \cite{bach1994generalized} defined by
\begin{align}
    \label{eq:pdm}
    \Gamma(\k) &\equiv \frac{1}{e^{\beta H(\k)}+1} \\
               &= \frac{1}{2} - H(\k) U(\k)^\dagger \frac{1}{2|\sE(\k)|} \tanh \left( \frac{\beta |\sE(\k)|}{2} \right) U(\k)
\end{align}
Where $U(\k)$ is the unitary matrix which diagonalizes $H(\k)^2$ into $\sE(k)^2$. 
Notice that the 1-pdm $\Gamma(k)$ is related to the Nambu spinors in Eq. \eqref{eq:reduced-Nambu} in the sense that the matrix entries of $\Gamma(\k)$ are given by
\begin{equation}
    \Gamma(\k)_{ab} = \bra \Psi^\dagger(\k)_b \Psi(\k)_a\ket
\end{equation}
\subsubsection{A reformulation of $SU(2)\to SO(3)$}
The single-band BCS Hamiltonian is generally diagonalized via the $SU(2)\to SO(3)$ relation (or $\mathfrak{su}(2) \leftrightarrow \mathfrak{so}(3)$ isomorphism) and thus it is natural to think the corresponding $SU(4)\to SO(6)$ homomorphism ($\mathfrak{su}(4) \leftrightarrow \mathfrak{so}(6)$) can be used to diagonalize the 2-band BCS Hamiltonian.
With that in mind, let us reformulate the $SU(2)\to SO(3)$ relation so that it can be easily generalized to the $SU(4) \to SO(6)$ isomorphism.

Let $\hat{a}\wedge \hat{b}$ be the anti-symmetric matrix with entry $-1$ at $(a,b)$ and $+1$ as $(b,a)$. 
In the example of $d=3$ dimensions, 
\begin{align}
    \hat{1} \wedge \hat{2} &= -i L_3 \\
    \exp( \theta  \hat{1} \wedge \hat{2}) &= \exp (-i\theta L_3)
\end{align}
Where $L_3$ is the standard angular momentum along the $+\hat{3}$ direction.
Notice that in $d=3$ dimension, a rotation in a plan  (e.g., $e^{-i\theta L_3}$) can be specified by the normal vector $\hat{3}$ and a rotation angle $\theta$.
However, in higher dimensions, this is impossible and thus we instead specifiy a rotation by the plane, e.g., $\hat{1}\wedge \hat{2}$ (the plane spanned by the orthonormal frame $\hat{1},\hat{2}$), and the rotation angle $\theta$. Notice that $\hat{1}\wedge{2} = -\hat{2}\wedge \hat{1}$, and thus the order of the 2 axes represent the direction of rotation, e.g., $\exp(\theta \hat{1}\wedge \hat{2})$ rotates the axes $\hat{1} \to \hat{2}$ by an angle of $\theta$.

It's then clear that the $SU(2)\to SO(3)$ relation ($\mathfrak{su}(2)\leftrightarrow \mathfrak{so}(3)$ isomorphism) is given by
\begin{align}
    \frac{1}{2i} \sigma_1 \leftrightarrow \hat{2}\wedge \hat{3} \\
    \frac{1}{2i} \sigma_2 \leftrightarrow \hat{3}\wedge \hat{1} \\
    \frac{1}{2i} \sigma_3 \leftrightarrow \hat{1}\wedge \hat{2}
\end{align}
Where $\sigma_i$ are the Pauli matrices. In particular, a $2\times 2$ Hermitian matrix in the Lie algebra $i\mathfrak{su}(2)$
\begin{equation}
    \label{eq:SU2-example}
    H = \cos \theta \sigma_3 + \sin \theta \sigma_1 \leftrightarrow \cos \theta \hat{1}\wedge \hat{2} + \sin \theta \hat{2}\wedge \hat{3}
\end{equation}
Can be diagonalized $H = U^\dagger \sigma_3 U$ via the unitary operator
\begin{equation}
    U = \exp \left(-\theta \frac{\sigma_2}{2i} \right) \leftrightarrow \exp\left(\theta \hat{3}\wedge\hat{1} \right)
\end{equation}
It should be emphasized that the only requirement for such a diagonalization procedure is that $\hat{1},\hat{2},\hat{3}$ form an orthonormal frame.
In higher dimensions, such an orthonormal frame may be arbitrarily chosen, and may not correspond to the original standard axes.
\subsubsection{The $SU(4)\to SO(6)$ Relation and its Application to the 2-band System}
\begin{table}[h!]
\centering
\begin{equation}
\begin{array}{c||c|ccc}
  & 0                     & 1                                    & 2                              & 3\\\hline\hline
0 &                       & \tm{l1}\hat{2}\wedge\hat{3}\tm{r1}   & \hat{3}\wedge\hat{1}           & \tm{l2}\hat{1}\wedge \hat{2}\tm{r2} \\\hline
1 & \hat{2}'\wedge\hat{3}'& \hat{1}'\wedge\hat{1}                & \tm{l3}\hat{1}'\wedge\hat{2}   & \hat{1}'\wedge\hat{3}\\
2 & \hat{3}'\wedge\hat{1}'& \hat{2}'\wedge\hat{1}                & \hat{2}'\wedge\hat{2}          & \hat{2}'\wedge\hat{3}\\
3 & \hat{1}'\wedge\hat{2}'& \hat{3}'\wedge\hat{1}                & \hat{3}'\wedge\hat{2}\tm{r3}   & \hat{3}'\wedge\hat{3}
\end{array}
\DrawBox[thick, red ]{l1}{r1}{0.2}{0.2}
\DrawBox[thick, red ]{l2}{r2}{0.2}{0.2}
\DrawBox[thick, red ]{l3}{r3}{0.2}{0.2}
\end{equation}
\label{tab:SU4}
\caption{$SU(4)\to SO(6)$ Relation. $\hat{1},\hat{2},\hat{3},\hat{1}',\hat{2}',\hat{3}'$ denote the 6-axes generating $\mathbb{R}^6$ so that $\hat{a}\wedge\hat{b}$ generates the Lie algebra $\mathfrak{so}(6)$. The entry $\hat{a}\wedge\hat{b}$ in row $\mu$ and column $\nu$ corresponds to $\sigma_{\mu\nu}/2i$ where $\sigma_{\mu\nu}=\sigma_\mu\otimes \sigma_\nu$ is the tensor product of Pauli matrices, e.g., $\sigma_{13}/2i \leftrightarrow \hat{1}'\wedge\hat{3}$. The red box show the only nonzero entries in $H(\k)^2$ in Eq. \eqref{eq:H-square}.}
\end{table}

Using the framework established in the previous section, we can easily generalize the isomorphism to $SU(4) \to SO(6)$, tabulated in Table. \ref{tab:SU4}. 
Indeed, let $\hat{1},\hat{2},\hat{3},\hat{1}',\hat{2}',\hat{3}'$ denote the 6-axes generating $\mathbb{R}^6$ so that $\hat{a}\wedge\hat{b}$ generates the Lie algebra $\mathfrak{so}(6)$.
Let $\sigma_{\mu\nu} =\sigma_\mu \otimes \sigma_\nu$ denote the tensor product of Pauli matrices (all except $\sigma_{00}$) which generates the Lie algebra $i\mathfrak{su}(4)$.
Then Table. \ref{tab:SU4} is read in the following manner: the entry $\hat{a}\wedge\hat{b}$ in row $\mu$ and column $\nu$ corresponds to $\sigma_{\mu\nu}/2i$, e.g., $\sigma_{13}/2i \leftrightarrow \hat{1}'\wedge\hat{3}$.

To apply this isomorphism in diagonalizing the 2-band system, notice that the red lines in Table. \ref{tab:SU4} boxes the only entries occurring in $H(\k)^2$ as given in Eq. \eqref{eq:H-square}. 
Therefore, we can map $H(\k)^2$ into the Lie algebra $\mathfrak{so}(6)$ so that
\begin{align}
    H(\k)^2 &\leftrightarrow h \hat{1}\wedge \hat{2} + g \hat{2}\wedge \hat{n} \\
    &\leftrightarrow r \left(\cos \theta \hat{1}\wedge \hat{2} + \sin \theta \hat{2}\wedge \hat{n} \right)
\end{align}
Where we have ignored that constant $R$ term since is commutes with any unitary operator and
\begin{align}
    g\hat{n} = a\hat{3} +b\hat{3}'+c\hat{2}'+d\hat{1}', \quad r^2 =h^2+g^2
\end{align}
It's then clear that $\hat{1},\hat{2},\hat{n}$ form an orthonormal frame in $\mathbb{R}^6$, and thus using the example given in Eq. \eqref{eq:SU2-example}, it's intuitive to see that $H(\k)^2$ is diagonalized via
\begin{align}
    U(\k) &\leftrightarrow \exp \left( \theta \hat{n}\wedge \hat{1} \right) \\
     &= \exp \left(i \frac{\theta}{2g} (a\sigma_{02} +b\sigma_{31} +c\sigma_{21} +d\sigma_{11}) \right) \\
     &= \cos \left(\frac{\theta}{2}\right) +i \sin \left(\frac{\theta}{2}\right) \frac{1}{g} (a\sigma_{02} +b\sigma_{31} +c\sigma_{21} +d\sigma_{11})
\end{align}
So that
\begin{equation}
    H(\k)^2 = U(\k)^\dagger \sE(\k)^2 U(\k), \quad \sE(\k)^2 = R +r\sigma_{03}
\end{equation}
In particular, there exists $F(\k),f(\k)$ given by
\begin{align}
    F(\k) &= \sum_{\sigma =\pm} \frac{1}{4\sqrt{R(\k)+\sigma r(\k)}} \tanh\left( \frac{\beta}{2} \sqrt{R(\k)+ \sigma r(\k)}\right) \\
    f(\k) &= \sum_{\sigma =\pm} \frac{\sigma}{4\sqrt{R(\k)+\sigma r(\k)}} \tanh\left( \frac{\beta}{2} \sqrt{R(\k)+ \sigma r(\k)} \right)
\end{align}
Such that
\begin{align}
    \frac{1}{2|\sE(\k)|} \tanh \left( \frac{\beta |\sE(\k)|}{2} \right) &= F(\k) +f(\k) \sigma_{03} \\
    U(\k)^\dagger \frac{1}{2|\sE(\k)|} \tanh \left( \frac{\beta |\sE(\k)|}{2} \right) U(\k) &=
    F(\k)  + \frac{f(\k)}{r(\k)}
    \left(
    \begin{array}{*2c|*2c}
      h        & a+ib        &  0       & c+id   \\[0.5ex]
    a-ib       & -h          &  -c-id    & 0      \\[0.5ex] \hline
    0          & -c+id       &  +h      & a-ib   \\[0.5ex]
    c-id       & 0           &  a+ib    & -h    \\[0.5ex]
    \end{array}\right) \\
    &= F(\k) + \frac{f(\k)}{r(\k)} (H(\k)^2 -R(\k))
\end{align}
It's then easy to obtain the 1-pdm $\Gamma(\k)$ in Eq. \eqref{eq:pdm}.
\subsection{Simulating strain and stress}
\label{sec:strain-details}
In this section, we describe the detailed parameters used to simulate pure $A_1$ and $B_2$ strain in the main text. For simplicity, let us define
\begin{align}
    \ve_0 &\equiv \frac{1}{2} (\ve_x +\ve_y) \\
    \ve_3 &\equiv \frac{1}{2} (\ve_x -\ve_y) 
\end{align}
Where $\ve_0,\ve_3$ should be regarded as the coefficients of decomposing the $x,y$ band structure into Pauli matrices, i.e.,
\begin{equation}
    \left(\begin{array}{cc}
    \ve_x & \ve_h \\
    \ve_h & \ve_y
    \end{array}\right) 
    = \ve_0 \sigma_0 +\ve_h \sigma_1 +\ve_3\sigma_3
\end{equation}
In the absence of strain, $\ve_x,\ve_y$ are given by Eq. \eqref{eq:band-structure-details} and thus we rewritten the equations as follows
\begin{align}
    \label{eq:band-structure-Pauli}
    \ve_0(\k)  &= -\mu_0 -2t_0(\cos k_x +\cos k_y) +4t_0' \sin k_x \sin k_y,\quad t_0 = (t_x+t_y)/2 \\
    \ve_3 (\k) &= - 2t_3 (\cos k_x - \cos k_y)-2t_3' (\cos k_x -\cos k_y),\quad t_3 = (t_x-t_y)/2 \\
    \ve_z (\k) &= -\mu_z -2t_z (\cos k_x +\cos k_y) -2t_z'(2 \cos k_x \cos k_y) \\
    \hy (\k)   &= 4 t_h \sin k_x \sin k_y
\end{align}
To simulate different types of strain, let us introduce the following notation
\begin{align}
    \label{eq:band-structure-strain}
    \ve_0(\k)  &= -\mu_{0,A_1} -2t_{0,A_1}(\cos k_x +\cos k_y) -2t_{0,B_1} (\cos k_x -\cos k_y) +4t_{0,B_2} \sin k_x \sin k_y \\
    \ve_3 (\k) &= - 2t_{3,A_1} (\cos k_x - \cos k_y)-2t_{3,B_1} (\cos k_x +\cos k_y) \\
    \ve_z (\k) &= -\mu_{z,A_1} -2t_{z,A_1} (\cos k_x +\cos k_y) -2t_{z,A_1}'(2 \cos k_x \cos k_y) \\
               &\quad -2t_{z,B_1}(\cos k_x -\cos k_y) +4t_{z,B_2} \sin k_x \sin k_y \\
    \hy (\k)   &= 4 t_{h,A_1} \sin k_x \sin k_y 
\end{align}
To simulate pure $A_1$ strain $\ve_{A_1}$, we have
\begin{align}
    t_{s,A_1} (\ve_{A_1})&= t_s (1-\alpha_{t,A_1} \ve_{A_1}), \quad s=z,z',0,3,h \\
    \mu_{z,A_1}(\ve_{A_1}) &= \mu_0 (1-\alpha_{\mu,A_1} \ve_{A_1})
\end{align}
Where $\mu_{0,A_1}$ is modified so that the total density of electrons $\bra n\ket$ is kept fixed for nonzero values of $\ve_{A_1}$, and we have taken $\alpha_{t,A_1} = 3$ and $\alpha_{\mu,A_1}=8$ so that we mainly modify the chemical potentials when simulating pure $A_1$ strain. Similarly, to simulate pure $B_1$ and $B_2$ strain, we have
\begin{align}
    t_{s,B_1}(\ve_{B_1}) &= -t_s\alpha_{B_1} \ve_{B_1}, \quad s=z,0,3 \\
    t_{s,B_2}(\ve_{B_2}) &= -t_3 \alpha_{B_2} \ve_{B_2}, \quad s=z,0
\end{align}
Where we have taken $\alpha_{B_1}=\alpha_{B_2}=12$.

It should be noted that in order to simulate uniaxial stress, we used the Poisson ratio $\nu\approx 0.3935$ for \sro{} \cite{paglione2002elastic} so that both $A_1$ strain and $B_1$ strain are present, i.e.,
\begin{align}
    \ve_{yy} &= -\nu \ve_{xx} \\
    \ve_{A_1} &= \frac{1}{2}(\ve_{xx} +\ve_{yy}) \\
    \ve_{B_1} &= \frac{1}{2}(\ve_{xx} -\ve_{yy})
\end{align}
The magnitudes $\alpha_{t,A_1},\alpha_{\mu,A_1},\alpha_{B_1}$ were tuned \footnote{Notice that the magnitudes $\alpha$ are slightly different from the SI of reference \cite{li2022elastocaloric}. Indeed, the reference only considered the single $\gamma$ band case and did not modify the chemical potential $\mu_z$, which we believe to be the main source of $A_1$ strain since it first appears as an onsite hopping term.} so that the $\gamma$ band in \sro{} crosses the Van Hove point at $\ve_{xx}\approx -0.44\%$ \cite{li2022elastocaloric}, while the $\alpha,\beta$ bands distort slightly (an asymmetry of $\sim 2\%$) \cite{sunko2019direct}. 
\end{document}